\documentclass[twocolumn, prx,superscriptaddress]{revtex4-2}

\usepackage{amsmath, amssymb, amstext, amsfonts}
\usepackage{amsthm}
\usepackage{float}
\usepackage{graphicx}
\usepackage{physics}
\usepackage{xcolor}
\usepackage{lineno}

\def\R{\mathbb{R}}
\def\S{{\cal S}}
\def\P{{\cal P}}
\def\App{Appendix }

\newtheorem{theo}{Theorem}

\newtheorem{lemma}[theo]{Lemma}

\newtheorem{problem}{Problem}

\def\id{{\mathbb I}}

\def\braket#1#2{\langle#1|#2\rangle}

\def\be{\begin{equation}}
\def\ee{\end{equation}}
\def\bea{\begin{eqnarray}}
\def\eea{\end{eqnarray}}
\def\bma{\begin{mathletters}}
\def\ema{\end{mathletters}}

\def\P{{\cal P}}

\def\q0{\underline{0}}

\def\H{{\cal H}}

\def\P{{\cal P}}
\def\S{{\cal S}}
\def\C{{\mathbb C}}
\def\id{{\mathbb I}}

\def\M{{\cal M}}
\def\H{{\cal H}}

\def\R{\mathbb{R}}

\def\App{Appendix\;}

\begin{document}

\title{Quantum advantages for transportation tasks: projectiles, rockets and quantum backflow}
\author{David Trillo} 
\affiliation{Institute for Quantum Optics and Quantum Information (IQOQI) Vienna, Boltzmanngasse 3, A-1090 Vienna}
\affiliation{University of Vienna, Faculty of Physics \& Vienna Doctoral School in Physics,  Boltzmanngasse 5, A-1090 Vienna, Austria}
\author{Thinh P. Le}
\author{Miguel Navascu{\'e}s}
\affiliation{Institute for Quantum Optics and Quantum Information (IQOQI) Vienna, Boltzmanngasse 3, A-1090 Vienna}
\begin{abstract}

Consider a scenario where a quantum particle is initially prepared in some bounded region of space and left to propagate freely. After some time, we verify if the particle has reached some distant target region. We find that there exist `ultrafast' (`ultraslow') quantum states, whose probability of arrival is greater (smaller) than that of any classical particle prepared in the same region with the same momentum distribution. For both projectiles and rockets, we prove that the quantum advantage, quantified by the difference between the quantum and optimal classical arrival probabilities, is limited by the Bracken-Melloy constant $c_{bm}$, originally introduced to study the phenomenon of quantum backflow. In this regard, we substantiate the $29$-year-old conjecture that $c_{bm}\approx 0.038$ by proving the bounds $0.0315\leq c_{bm}\leq 0.072$. Finally, we show that, in a modified projectile scenario where the initial position distribution of the particle is also fixed, the quantum advantage can reach $0.1262$.

\end{abstract}

\maketitle

\section{Introduction}

Much of current research in quantum theory focuses on the exploitation of quantum effects in communication and computation. Nevertheless, quantum systems are originally found to be advantageous for mechanical tasks. A paradigmatic example is the tunneling effect \cite{tunneling1}: A quantum particle can be detected in regions of space that are classically forbidden by energy considerations. Another  noteworthy example is quantum backflow: A free quantum particle with positive momentum can be observed to propagate backwards. Quantum backflow was first identified by Allcock in the context of the time-of-arrival problem \cite{Allcock}, and later isolated by Bracken and Melloy \cite{BM}. More recent examples of quantum advantage in mechanical systems can be found in \cite{tsirelson} and \cite{valerio}.

The advantages that quantum mechanical systems might offer for transportation, understood as the quick dispatch of massive particles through free space, are, however, unexplored. Some effort has been paid to investigate the properties of a hypothetical quantum time-of-arrival operator \cite{muga2007time} in connection with quantum backflow. Perhaps due to its foundational character, this research program has not produced so far any concrete task where quantum mechanical systems have the upper hand.

In this work, we prove the advantage of quantum mechanical systems over their classical counterparts in a practical transportation task, which we call the {\em projectile scenario}. Consider a situation where a non-relativistic one-dimensional quantum particle (a projectile) is prepared in some bounded region of space $B$ and left to propagate freely. After some time $\Delta T$, we measure if the particle is in some distant target region $R$. For a fixed initial quantum state $\rho$ with spatial support in $B$, we compare the probability of detection in $R$ with that of a classical particle, initially prepared in $B$ with the same momentum distribution as $\rho$. 

We find that there exist what one might call \emph{ultra-fast states} (\emph{ultra-slow states}), whose probability of detection in $R$ at time $\Delta T$ is strictly greater (smaller) than that of any classical particle. A natural figure of merit for quantum advantage in the ultra-fast regime is the difference between the quantum and the maximum classical probabilities of arrival. Likewise, in the ultra-slow regime one can consider the difference between the minimum classical and the quantum probabilities of arrival. We find that the maximum quantum advantage in either case does not depend on the distance between the preparation and target regions, but only on the parameter $\alpha := M|B|^2/\Delta T$. For finite values of $\alpha$, the maximum quantum-classical gap can be computed up to precision $\delta$ by diagonalizing an $N\times N$ matrix, with $N=O\left(\log\left(1/\delta\right)\right)$.

We prove that the maximum quantum advantage, achieved in the limit $\alpha \rightarrow \infty$, equals the Bracken-Melloy constant \cite{BM}, which was numerically estimated to have the value $c_{bm} \approx 0.0384517$ \cite{backflow1, backflowgood}. This conjectured value was, however, not computed with any rigorous error bounds. In fact, until now there was no reason to believe that $c_{bm}$ was smaller than $1$. In this regard, we argue that $0.0315 \leq c_{bm} \leq 0.0725$, hence providing the first upper bound on $c_{bm}$.

As we show, the appearance of $c_{bm}$ is not a coincidence: through simple metaplectic transformations we connect the quantum projectile problem with a variety of scenarios related to and generalizing quantum backflow, including quantum backflow itself. All such effects are therefore manifestations of the same mathematical phenomenon, seen through different coordinate systems. In the light of the recent interest in experimentally demonstrating quantum backflow \cite{exp2, exp, backflowfriend, backflowexperimentfriend, classical_experiment}, we argue that projectile scenarios are more experimentally friendly and operationally interesting.

To arrive at a transportation task with a quantum advantage beyond the Bracken-Melloy constant, we consider a scenario in which several projectiles are sequentially released, namely, a quantum rocket. However, it turns out that $c_{bm}$ also limits the advantage of a quantum rocket over a classical analog with the same lift-off zone, combustion chamber size and rocket and fuel momentum distributions. 

Nevertheless, we show that a superior quantum advantage can actually be attained in a variant of the projectile scenario where the quantum projectile is compared with a classical particle having the same position and momentum distributions.

The paper is structured as follows: in section \ref{class_vs_quant_proj} we introduce and solve the projectile scenario; the connection between quantum projectiles and other examples of quantum advantage in mechanical systems is explained in section \ref{connection_sec}. In section \ref{rocket_sec} we provide a simple model for quantum rockets and use it to prove that the classical-quantum gap in such artifacts is also limited by the Bracken-Melloy constant. In section \ref{variant_sec}, we add a natural constraint to the projectile scenario so that the Bracken-Melloy limit can be superseded. Finally, in section \ref{conclusion_sec} we present our conclusions. We also provide some Appendices in which the lengthier computations are made more explicit.

\section{Classical vs. quantum projectiles}
\label{class_vs_quant_proj}

\begin{figure*}
    \centering
    \includegraphics[width=1\linewidth]{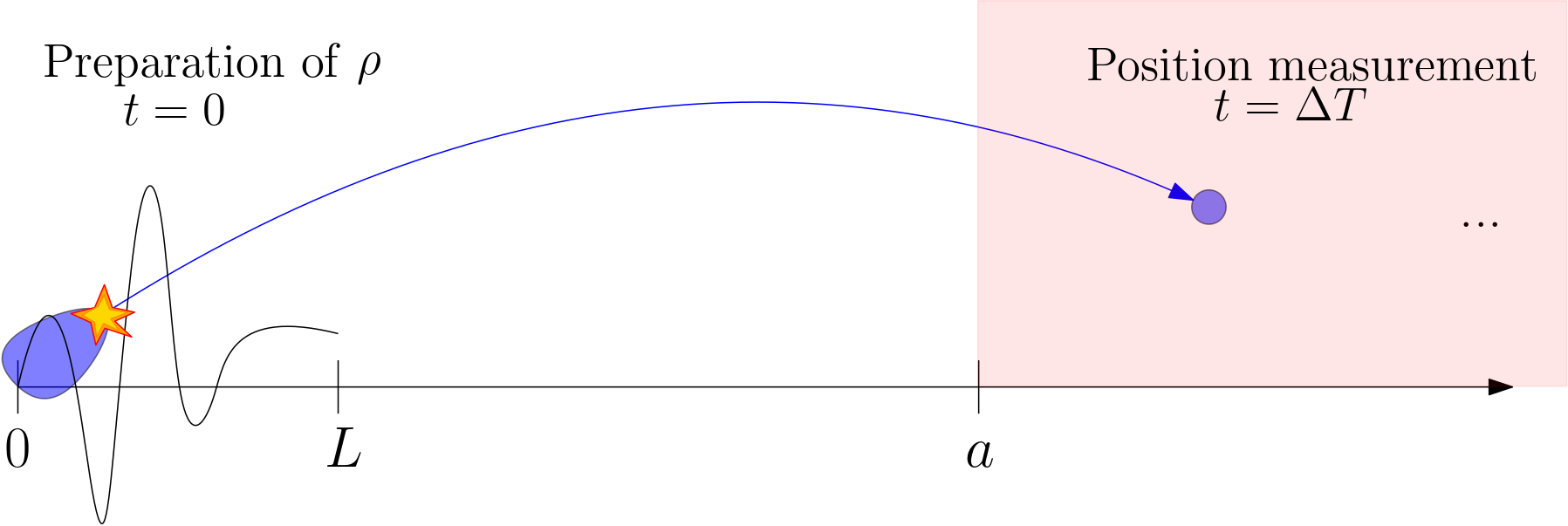}
    \caption{Projectile scenario. A projectile is prepared at time $t=0$ in $[0,L]$ and, at time $t=\Delta T$, we verify that it has reached region $[a,\infty)$. Maximum quantum advantage in probability of arrival as compared to a classical particle is found to be the  Bracken-Melloy constant, $0.0315 \leq c_{bm} \approx 0.0384517 \leq 0.0725$.}
    \label{fig: setup}
\end{figure*}

Our starting point is a classical projectile of mass $M$, prepared at time $t=0$ in the region $[0,L]$. At time $t=\Delta T>0$, we observe whether the projectile has reached region $[a,\infty)$, with $a>L$ (see Figure \ref{fig: setup}). If we ignore where exactly in $[0,L]$ the projectile was prepared, then the probability of finding it in $[a,\infty)$ at time $\Delta T$ is, at most, $\mbox{Prob}\left(p\geq M(a-L)/\Delta T\right)$, where $p$ denotes the projectile's linear momentum. This corresponds to a configuration where the projectile was prepared at $x=L$ at time $t=0$. Similarly, the probability to find the projectile in $[a,\infty)$ at time $\Delta T$ is, at least, $\mbox{Prob}\left(p\geq Ma/\Delta T\right)$, which corresponds to an initial preparation at $x=0$.

Now, let us assume that the projectile is, in fact, a quantum mechanical system. Let ${\cal S}(R)$ denote the set of quantum states with spatial support in $R \subset \R$. We will omit the parentheses whenever $R$ is an interval, and thus denote by $\rho\in \S[0,L]$ the initial quantum state of the projectile. While the projectile is freely propagating, its dynamics are governed by the kinetic Hamiltonian $H=P^2/2M$, where $P$ denotes the projectile's linear momentum operator. The probability to find the quantum projectile in region $[a,\infty)$ after time $\Delta T$ can be found by simple application of the Born rule: it is $\tr(U\rho U^\dagger\Theta(X-a))$, where $U:=e^{-iH\Delta T}$ and $\Theta$ is the Heaviside step function. Note that we work in units where $\hbar=1$.

If, after time $\Delta T$, the quantum projectile is found in $[a,\infty)$ with probability greater than any classical particle initially prepared in $[0,L]$ with the same momentum distribution, we say that the quantum projectile is \emph{ultra-fast}. If, on the contrary, the projectile is detected with probability lower than the classical minimum, we say that the projectile is \emph{ultra-slow}. To gauge how ultra-fast or ultra-slow a quantum projectile in state $\rho$ is, we consider the difference between the quantum and optimal classical probabilities of arrival.

Let us deal with the ultrafast case first. As we saw in the first paragraph of this section, a classical projectile with momentum distribution $\nu(p)dp$ will be detected in $[a, \infty)$ at time $\Delta T$ with probability at most $\mbox{Prob}\left(p\geq M(a-L)/\Delta T\right)$. The probability of this event is to be evaluated on the distribution $\nu(p)dp$. Since we have assumed $\nu(p)dp$ to be the same as the momentum distribution of a quantum particle in state $\rho$, this implies
\[
\mbox{Prob}\left(p\geq \frac{M(a-L)}{\Delta T}\right) = \tr\left[\rho \Theta\left(\frac{\Delta T}{M}P - (a-L)\right)\right].
\]
Thus the quantum advantage, if it exists, is given by $\tr(\rho \Omega_F(M,a,\Delta T))$, with
\[
\Omega_F(M,a,\Delta T) := \Theta\left(X+\frac{\Delta T}{M}P-a\right) - \Theta\left(\frac{\Delta T}{M}P-a+L\right),
\]
where, in the first term of the right-hand side, we made use of the identity \footnote{The identity is a consequence of the formulas $\frac{dX}{dt}=i[H,X]=\frac{P}{2M}$, $\frac{dP}{dt}=i[H,P]=0$.} $U^\dagger X U=X+\Delta TP/M$.

We wish to find the largest advantage achievable with a quantum projectile. That is, we are interested in the quantity
\[
\varphi_F(M,L,a,\Delta T) := \sup_{\rho \in \S[0,L]} \tr(\rho \Omega_F(M,a,\Delta T)).
\] 
Given a set of states $S$ and an operator $A$, we have, for any unitary $U$, that 
\[
\sup_{\rho \in S} \tr(\rho A) = \sup_{\rho \in USU^\dagger} \tr(\rho UAU^\dagger).
\]
We next exploit this observation to prove that $\varphi_F$ is just a function of $\alpha := ML^2/\Delta T$. In particular, $\varphi_F$ does not depend on $a$, the location of the target region: remarkably, quantum projectiles are equally advantageous no matter how large the flight distance.

Let $\sigma : \R^2 \rightarrow \R^2$ be an affine linear transformation and consider the vector of operators $(X,P)$. If $\sigma$ is metaplectic, namely $[\sigma(X,P)_1,\sigma(X,P)_2]=[X,P]=i$, then, as we show in \App A, there exists a unitary $U_\sigma$ such that
\begin{equation}
 \left(U_\sigma XU^\dagger_\sigma,U_\sigma P U_\sigma^\dagger\right) =\sigma(X,P).
\end{equation}

Now, consider the unitary $V$ associated to the metaplectic map
\begin{equation}
\begin{aligned}
x & \longmapsto \sqrt{\frac{M}{\Delta T}}(x-L),\\
p & \longmapsto \sqrt{\frac{\Delta T}{M}}p-\sqrt{\frac{M}{\Delta T}}(a-L).
\end{aligned}
\label{ultra2stand}
\end{equation}
For $\alpha=ML^2/\Delta T$, it follows that
\begin{gather*}
V\S[-\sqrt{\alpha},0]V^\dagger=\S[0,L],\\
V\Omega_F(M,a,\Delta T)V^\dagger=\Theta(X+P)-\Theta(P)=:\Omega,
\end{gather*}
therefore
\begin{equation}
\varphi_F(M,L,a,\Delta T)=\varphi(\alpha):= \sup_{\rho \in \S[-\sqrt{\alpha},0]} \tr(\rho \Omega).
\label{standard}    
\end{equation}
\noindent Hence, $\varphi_F$ is just a function of $\alpha$. We call the right-hand side of the above equation the standard projectile problem, or \emph{standard problem} for short. Note that the standard problem corresponds to determining the maximum quantum advantage of an ultrafast projectile of mass $M=1$, prepared in the region $[-\sqrt{\alpha},0]$, to be found in region $[0,\infty)$ after time $\Delta T = 1$. 

So far we have only considered ultrafast projectiles. For the ultraslow case, the story is pretty much the same, but opposite: namely, we are now interested in \emph{not} finding the particle in the target region $[a,\infty)$ after time $\Delta T$ has elapsed. The optimal classical strategy is now to concentrate all the mass at point $x=0$. In this case, the probability that a classical projectile, prepared at time $t=0$ in $[0,L]$ with the same momentum distribution as the quantum state $\rho$, reaches the target region at time $t=\Delta T$ is given by
\[
\mbox{Prob}\left(p\geq \frac{Ma}{\Delta T}\right) = \tr\left[\rho \Theta\left(\frac{\Delta T}{M}P - a\right)\right],
\]
and so the quantum advantage, if it exists, of not finding the particle in the target region is quantified by $\tr(\rho \Omega_S(M,a,\Delta T))$, with
\[
\Omega_S(M,a,\Delta T) := \Theta\left(\frac{\Delta T}{M}P-a\right)-\Theta\left(X+\frac{\Delta T}{M}P-a\right).
\]
The maximum quantum advantage is thus 
\[
\varphi_S(M,L,a,\Delta T) := \sup_{\rho \in \S[0,L]} \tr(\rho \Omega_S(M,a,\Delta T)).
\]

As it turns out, $\varphi_S=\varphi$, and so the functions $\varphi_F$, $\varphi_S$ are identical. Indeed, consider the transformation 
\begin{equation}
\sigma(x,p)= \left(-\sqrt{\frac{M}{\Delta T}}x, \sqrt{\frac{\Delta T}{M}}p+\sqrt{\frac{M}{\Delta T}}(x-a)\right). 
\end{equation}
Since $[\sigma(X,P)_1,\sigma(X,P)_2]=-i$, this map does not define a unitary transformation over the set of quantum states. Rather, it defines an \emph{anti-unitary} transformation $U_\sigma$, as explained in \App A.  Now, the argument above relating linear optimizations over subsets of quantum states also extends to anti-unitary transformations. The reader can verify that, applying $U_\sigma$ to the standard problem with $\alpha=ML^2/\Delta T$, one ends up with the definition of $\varphi_S$, and, therefore, $\varphi_S(M,L,a,\Delta T)=\varphi\left(ML^2/\Delta T\right)$.

In section \ref{solving_standard_sec}, we will prove that $\varphi(\alpha)>0$ for all $\alpha>0$, i.e., there exist ultrafast and ultraslow quantum states in any projectile scenario. From eq. (\ref{standard}) it is clear that $\varphi(\alpha)$ is a non-decreasing function. Moreover, as shown in section \ref{connection_sec}, its limiting (supremum) value $\varphi(\infty)$ corresponds to the Bracken-Melloy constant $c_{bm}$ \cite{Bracken}, conjectured to have the value $0.0384517$. We conclude that quantum projectiles can exhibit a limited advantage with respect to their classical counterparts.

We finish this section by introducing yet another projectile scenario. As before, we wish the quantum projectile to have a larger probability of arrival, but this time we award some advantage to the classical projectile: namely, we compare the probability to detect the quantum projectile in the region $[a,\infty)$ with the maximum probability of detecting the classical one in the larger region $[a-b,\infty)$, with $b>0$. This problem can be reduced, via the transformation (\ref{ultra2stand}), to an optimization of $\left\langle\Theta(X+P)-\Theta(P+\beta)\right\rangle_\rho$ over $\rho\in\S[-\sqrt{\alpha},0]$, with $\alpha=ML^2/\Delta T$, $\beta=b\sqrt{M/\Delta T}$. We denote this problem the \emph{extended standard problem}, with solution $\varphi(\alpha,\beta)$. Clearly, $\varphi(\alpha,\beta)$ is non-increasing in $\beta$ and $\varphi(\alpha,0)=\varphi(\alpha)$. Obviously, $\lim_{\beta\to\infty}\varphi(\alpha,\beta)=0$, and so one cannot reduce the extended standard problem to the standard problem.

\subsection{Solving the standard problem}
\label{solving_standard_sec}
From the formulation of the standard problem (\ref{standard}), one can immediately deduce that $\varphi$ is a non-decreasing function of $\alpha \in [0,\infty)$, with $\varphi(0) = 0$ and $\varphi(\alpha) \leq 1$. It remains to see that $\varphi(\alpha)\neq 0$ for some $\alpha$. To do this, we need to study the spectrum of $\Omega:=\Theta(X+P)-\Theta(P)$ restricted to the space $\S[-\sqrt{\alpha},0]$. In \App B we prove that, in position representation,
\begin{equation} \label{eq: operator in position}
\Omega{\big|}_{\S[-\sqrt{\alpha},0]} = \frac{1}{2\pi} \int_{[-\sqrt{\alpha},0]^2} dx dy \frac{e^{\frac{i}{2}(y^2-x^2)}-1}{i(y-x)}\ketbra{x}{y}
\end{equation}

\begin{figure}[t!]
    \centering
    \includegraphics[width=0.5\textwidth]{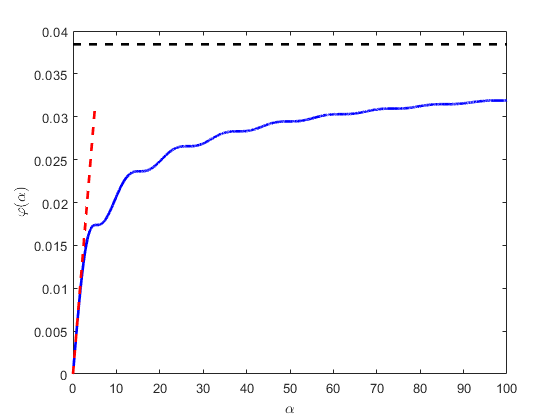}
    \caption{Solid blue: plot of $\varphi(\alpha)$ for $\alpha \in [0,100]$, computed with precision $\delta = 10^{-4}$. Dashed red: linear upper bound $(2\sqrt{3}-3)\alpha/24\pi$. Dashed black: the conjectured value of the Bracken-Melloy constant $c_{bm}$}
    \label{fig: phiplot}
\end{figure}

Let $K(x,y)$ be the kernel of this integral operator. If $\alpha > 0$, then we can choose $z \in (-\sqrt{\alpha},0)$ such that $K(0,z) =  K(z,0)^* \neq 0$. Since $K(0,0)=0$, by the determinant criterion it follows that the $2\times 2$ matrix $\{K(x,y)\}_{x,y=0,z}$ is not negative semidefinite. In particular, it has a positive eigenvalue $\lambda$, with eigenvector $(c_0,c_z)^T$. Now, consider the ket
\[
\ket{\psi_\varepsilon} = \frac{1}{\sqrt{\varepsilon}} \int_{[-\sqrt{\alpha},0]} dx (c_0 \chi_{[-\varepsilon,0]}(x) + c_z \chi_{[z-\varepsilon,z]}(x)) \ket{x},
\]
where $\chi_C$ denotes the characteristic function of $C\subset \R$. For small enough $\varepsilon$, $\ketbra{\psi_\epsilon}{\psi_\epsilon}\in \S[-\sqrt{\alpha},0]$ and $\bra{\psi_\varepsilon}\Omega\ket{\psi_\varepsilon} \approx \varepsilon \lambda >0$. We conclude that $\varphi(\alpha)>0$ for all $\alpha >0$, so ultrafast and ultraslow states exist in all projectile scenarios. 

The problem of computing $\varphi(\alpha)$ for different values of $\alpha$ is more convoluted. Note that the kernel $K(x,y)$ is analytic in $x,y$; hence, for $x,y\in [-\sqrt{\alpha},0]$, we can approximate it up to arbitrary precision by a polynomial on $x$ and $y$ of sufficiently high degree. When we replace $K(x,y)$ by its $N^{th}$ order Taylor expansion, we arrive at a new operator $\Omega_N$, which can be shown to be close in operator norm to $\Omega$, restricted to the subspace of wave functions defined in $[-\sqrt{\alpha},0]$. In turn, $\Omega_N$ only has support on the finite-dimensional subspace spanned by vectors of the form $\int_{[-\sqrt{\alpha},0]} dx x^k\ket{x}$, where $k$ runs from $0$ to the degree in $x$ of the kernel of $\Omega_N$. Hence $\Omega_N$ can be exactly diagonalized. In \App B this argument is developed to conclude that, for finite $\alpha$, we can compute $\varphi(\alpha)$ to any precision $\delta$ we want by diagonalizing a matrix of size $N \lesssim \max(\alpha, \log(1/\delta))$. In the same appendix, the reader can also find the following (tight) linear upper bound for $\varphi(\alpha)$:
\[
\varphi(\alpha) \leq \frac{2\sqrt{3}-3}{24\pi} \alpha.
\]
The function $\varphi(\alpha)$ is plotted for $\alpha\in [0,100]$ in Figure \ref{fig: phiplot}. As it can be appreciated, $\varphi(\alpha)$ roughly looks like a concave function, but not quite: at regular intervals, the slope of the function becomes very small. Such `steps' seem to decrease in amplitude as $\alpha$ grows, and, actually, for $\alpha\gg 1$, the function appears to be well approximated by the ansatz $r+s\alpha^{-1/2}$.

To grasp the maximum quantum advantage, we need to study the limiting case $\alpha = \infty$. The problem thus consists in determining the spectrum of $\Omega$, restricted to the space $L^2(-\infty, 0]$. To study this case, it is convenient to switch to the Wigner function representation.

The Wigner function of a quantum state $\rho$ is
\[
W_\rho (x,p):=\frac{1}{2\pi}\int_{-\infty}^{\infty} dy \bra{x-\frac{y}{2}}\rho \ket{x+\frac{y}{2}} e^{ipy}.
\]

For convenience, we recall the properties of Wigner functions in \App A. The most important one for us is the fact that Wigner functions behave nicely under metaplectic transformations in phase space. Namely, for any metaplectic transformation $U_\sigma$, it holds that
\begin{equation} \label{eq: wigner covariance}
W_\rho(\sigma^{-1}(x,p)) = W_{U_\sigma \rho U_\sigma^\dagger}(x,p).
\end{equation} 
Furthermore, for any bounded measurable function $f: \R \rightarrow \R$ and $a,b,c\in \R$, we have that
\begin{equation}
\tr(\rho f(aX+bP+c)) = \int_{\R^2} dx dp f(ax+bp+c) W_{\rho}(x,p),
\label{eq: f PS}
\end{equation}
where some care has to go into the precise meaning of the integral whenever the integrand is not Lebesgue integrable. Finally, note that, if $\rho$ has a convex support $R$ in either position or momentum, then the support of its Wigner function $W_\rho(x,p)$ corresponding to that variable is also contained in $R$.

Now, for any state $\rho$, we have, by eq. (\ref{eq: f PS}), that 
\[
\tr(\rho\Omega)=\int_{\R^2} dxdp W_\rho(x,p)(\Theta\left(x+p\right)-\Theta\left(p\right)).
\]
The last factor on the integrand will vanish everywhere, except in the regions $\Lambda^+=\{x+p\geq 0, p\leq 0\}$, where it equals $1$, and $\Lambda^-=\{x+p\leq 0, p\geq 0\}$, where it equals $-1$. However, if $\rho\in \S(-\infty, 0]$, then $W_\rho(x,p)=0$, for $x>0$. Since $(x,p)\in\Lambda^+$ implies $x\geq 0$, it follows that the first region does not contribute to the integration above. Hence,
\[
\varphi(\infty) = \sup_{\rho \in \S[-\infty,0]} -\int_{\Lambda^-} dx dp W_{\rho}(x,p).
\]
The problem of integrating Wigner functions over wedges (without any further constraints) was studied by Werner \cite{Werner_1988} in the context of time-of-arrival operators. The idea is that all wedges can be taken to each other via a metaplectic transformation, and therefore it suffices to study the wedge $[0,\infty) \times [0,\infty)$. Under this transformation, $\varphi(\infty)$ becomes
\[
 \sup_{\rho:\tr(\rho\Theta(X+P))=1} -\int_{[0,\infty)^2} dx dp W_{\rho}(x,p),
\]
where we have used that $\S(-\infty,0]$ is the space of states that satisfy the condition $\tr(\rho\Theta(-X))=1$. Werner considers the operator corresponding to integrating Wigner functions over the quadrant $x,p\geq 0$, and determines its spectrum to be $[-0.155940, 1.007678]$. Therefore, $\varphi(\infty) \leq 0.155940$. This bound, however, does not take into consideration the constraint $\tr(\rho \Theta(X+P))=1$. To account for it, we add to Werner's operator a linear combination of operators corresponding to integrating Wigner functions over hyperbolic regions in the quadrant $x,p\leq 0$. Since our Wigner functions vanish in that quadrant, the infimum of the spectrum of the new operator (which can also be determined with the techniques in \cite{Werner_1988}) also provides an upper bound for $\varphi(\infty)$. We numerically find the bound $\varphi(\infty) \leq 0.0725$, see \App D.


In addition, via variational methods, we show that $\varphi(\infty)\geq 0.0315$. This figure is obtained by optimizing linear combinations of the average values of the operators $\Omega$, $\Theta(X)$ over density matrices with support on the first $N+1$ number states $\{\ket{n}:n=0,...,N\}$, i.e., $(X+iP)\ket{n}=\sqrt{2n}\ket{n-1}$, see \App C for details. A plot of the Wigner function of a quantum state approximately in ${\cal S}(-\infty,0]$ and approximately achieving this value can be found in Figure \ref{fig: sameP} (left).

In the next section, we will show that $\varphi(\infty)=c_{bm}$, the Bracken-Melloy constant \cite{BM}, which is conjectured to have the value $0.0384517$ \cite{backflow1,backflowgood}. Our bounds $0.0315\leq c_{bm}\leq 0.0725$ therefore support this widespread belief.

\section{Connection with other quantum mechanical effects}
\label{connection_sec}

\begin{table*}[t!]
  \centering
  \begin{tabular}{|c|c|c||c|c|c|}
  \hline
    Scenario & Operator & Set of states & $\sigma(x)$ & $\sigma(p) $ & $\alpha$\\
    \hline
    Standard problem & \parbox{4cm}{\[\Theta\left(P+X\right)-\Theta\left(P\right)\]} & \parbox{1.5cm}{\[\S[-\sqrt{\alpha},0]\]} & $x$ & $p$ & $\alpha$\\
    \hline
    Ultrafast projectile &\parbox{6.25cm}{\[\Theta\left(X+\frac{\Delta T}{M}P-a\right)-\Theta\left(\frac{\Delta T}{M}P-(a-L)\right)\]} & \parbox{1cm}{\[\S[0,L]\]} & \parbox{2cm}{\[\sqrt{\frac{M}{\Delta T}}(x-L)\]} & \parbox{3.25cm}{\[\sqrt{\frac{\Delta T}{M}}p-\sqrt{\frac{M}{\Delta T}}(a-L)\]} & \parbox{1cm}{\[\frac{ML^2}{\Delta T}\]} \\
    \hline
    Ultraslow projectile & \parbox{5cm}{\[\Theta\left(\frac{\Delta T}{M}P-a\right)-\Theta\left(X+\frac{\Delta T}{M}P-a\right)\]} & \parbox{1cm}{\[\S[0,L]\]} & \parbox{1.5cm}{\[-\sqrt{\frac{M}{\Delta T}}x\]} & \parbox{3cm}{\[\sqrt{\frac{\Delta T}{M}}p+\sqrt{\frac{M}{\Delta T}}(x-a)\]} & \parbox{1cm}{\[\frac{ML^2}{\Delta T}\]}\\
    \hline
    Quantum backflow & \parbox{4cm}{\[\Theta\left(-X-\frac{\Delta T}{M}P\right)-\Theta\left(-X\right)\]}& \parbox{1cm}{\[\P[0,\infty)\]} & \parbox{1.5cm}{\[- \sqrt{\frac{\Delta T}{M}}p\]} & \parbox{1.5cm}{\[-\sqrt{\frac{M}{\Delta T}}x\]} & \parbox{1cm}{\[\infty\]}\\
    \hline
  \end{tabular}
  \caption{Most of the optimization problems considered in this paper are of the form $\max_{\rho\in S}\tr(\rho \Omega)$, for some operator $\Omega$ and some set of states $S$. This table contains the definitions of each problem and the reversible transformations mapping the standard problem to any other. $\S(R)$ denotes the set of states with position support in $R \subset \R$, and $\P(R)$ denotes the set of states with momentum support in $R \subset\R$. We use the shorthand $\sigma(x):= \sigma(x,p)_1$ and $\sigma(p):=\sigma(x,p)_2$, and omit parentheses whenever $R$ is an interval.}
  \label{table: relations}
\end{table*}

As we have seen, the ultrafast (ultraslow) projectile problem is equivalent to the standard problem, since a unitary (anti-unitary) transformation takes us from the latter to the former. We next see that the standard projectile problem is similarly connected to the most extreme manifestation of other quantum mechanical effects. The exact correspondences are summarized in Table \ref{table: relations}. The question of understanding the relation between some of these effects was raised in \cite{reentrybackflow} and partially answered in \cite{new_QB}. Our results answer the challenge posed in \cite{reentrybackflow} from a different point of view, namely, that of (anti-)unitary equivalence, and extend the connection to other mechanical effects.

Let us start with the phenomenon of quantum backflow \cite{Allcock,BM,backflow2,backflowgood}. Consider a pure state that only has positive momentum and that is evolving freely. In position representation, we can write it as
\[
\psi(x,t) = \frac{1}{\sqrt{2\pi}} \int_{0}^\infty dp e^{ipx}e^{-ip^2 t/2M} \phi(p).
\]
for some function $\phi$ such that $\int_0^\infty \abs{\phi(p)}^2 =1$. The probability flux at the origin is therefore
\[
j(0,t) = \frac{1}{4M\pi} \int_0^\infty dp dq (p+q) e^{it(q^2-p^2)/2M} \phi(p)\phi(q)^*,
\]
and thus the integrated flux at the origin from time $0$ to time $\Delta T$ is
\[
\int_0^{\Delta T} dt j(0,t) = \frac{1}{2\pi} \int_0^\infty dpdq \frac{e^{i \frac{\Delta T (q^2-p^2)}{2M}}-1}{i(q-p)} \phi(p)\phi(q)^*.
\]
Note the similarity with eq. (\ref{eq: operator in position}). Guided by classical intuition, one would expect this integrated flux to be non-negative, since the particle is only moving to the right. However, for some quantum states $\phi(x,t)$, this magnitude can be negative: in that case, we speak of quantum backflow. 

Alternatively, we can interpret quantum backflow as a \emph{decrease} in the probability of detecting a  particle with positive momentum in the region $[0,\infty)$. This is so because, by the continuity equation 
\[\frac{\partial}{\partial t}|\psi(x)|^2=-\frac{\partial}{\partial x}j(x,t),\]
the integrated flux satisfies:
\[
\int_0^{\Delta T} dt j(0,t) = \bra{\psi}U^\dagger \Theta(X) U\ket{\psi} - \bra{\psi}\Theta(X)\ket{\psi},
\]
where $\ket{\psi}=\int dx\psi(x,0)\ket{x}$ and $U=e^{-i\frac{P^2}{2M}\Delta T}$.

Call $\P[0,\infty)$ the space of all states with positive momentum support. From all the above it follows that the maximum amount of backflow is given by
\[
\sup_{\rho\in \P[0,\infty]} \tr(\rho \left(\Theta\left(-X-P\frac{\Delta T}{M}\right)-\Theta(-X)\right)):=c_{bm},
\]
\noindent where we used the identity $\Theta(z)=1-\Theta(-z)$. The number $c_{bm}$, known in the literature as the \emph{Bracken-Melloy constant} \cite{BM}, is thus the solution a problem of the form $\sup_{\rho \in S}\tr(\rho A)$, for some space of states $S$ and some operator $A$. In fact, this problem can be obtained from the standard problem with $\alpha = \infty$ via the anti-metaplectic transformation $\sigma(x,p)= (-p\sqrt{\Delta T/M}, -x\sqrt{M/\Delta T})$. Therefore, $c_{bm} = \varphi(\infty)$.

Going through the literature on quantum backflow, one finds that $c_{bm}$ is conjectured to have the value $0.0384517$ \cite{backflow1, backflowgood}. A figure of $0.038452$ is obtained in \cite{backflowgood} by fitting many points of (an approximation to) the graph of $\varphi(\alpha)$ with the ansatz $r-s\alpha^{-1/2}$ and, a figure of $0.0384517$ is obtained in \cite{backflow1}, by fitting such points to a degree $3$ polynomial over $\alpha^{-1/2}$. To our knowledge, prior to our work there were no rigorous, non-trivial upper bounds on $c_{bm}$, and the best lower bound fell $41\%$ short of the conjectured value of the constant \cite{Halliwell_2013}. Our results in the preceding section hence give mathematical support to the conjecture $c_{bm}\approx 0.0384517$.

\begin{table*}[t!]
  \centering
  \begin{tabular}{|c|c|c||c|c|c|}
  \hline
    Scenario & Operator & Set of states & $\sigma(x)$ & $\sigma(p) $ & $\beta$\\
    \hline
    \parbox{2cm}{Extended standard problem with $\alpha=\infty$} & \parbox{5.5cm}{\[\Theta\left(P+X\right)-\Theta\left(P\right)\]} & \parbox{1.5cm}{\[\S(-\infty,\beta]\]} & $x$ & $p$ & $\beta$\\
    \hline
    \parbox{2cm}{Generalized Quantum Backflow} & \parbox{5.5cm}{\[\Theta\left(-X-\frac{\Delta T}{M}P\right)-\Theta\left(-X\right)\]}& \parbox{1.5cm}{\[\P[-\gamma,\infty)\]} & \parbox{1.5cm}{\[- \sqrt{\frac{\Delta T}{M}}p\]} & \parbox{1.5cm}{\[-\sqrt{\frac{\Delta T}{M}}x\]} & \parbox{1.25cm}{\[\sqrt{\frac{\Delta T}{M}} \gamma\]} \\
    \hline
    \parbox{2cm}{Constant force QB} &\parbox{5.5cm}{\[\Theta\left(-X-\frac{\Delta T}{M}P+\frac{F\Delta T^2}{2M}\right)-\Theta\left(-X\right)\]} & \parbox{1cm}{\[\P[0,\infty)\]} & \parbox{3.25cm}{\[-\sqrt{\frac{\Delta T}{M}}\left(p-\frac{F\Delta T}{2}\right)\]} & \parbox{1.5cm}{\[-\sqrt{\frac{M}{\Delta T}} x\]} & \parbox{1.25cm}{\[\frac{F \Delta T }{2}\]} \\
    \hline
    \parbox{2cm}{Quantum reentry} & \parbox{5.5cm}{\[\Theta\left(l-X-\frac{t_2}{M}P\right)-\Theta\left(l-X- \frac{t_1}{M}P\right)\]} & \parbox{1.5cm}{\[\S(-\infty,0]\]} &  \parbox{2cm}{\[\sqrt{\frac{M C}{t_1}}\left(x-l\right)\]} & \parbox{3.25cm}{\[\sqrt{\frac{M}{t_1 C}}\left(l-x-\frac{t_1}{M}p\right)\]} & \parbox{1.25cm}{\[l\]}\\
    \hline
  \end{tabular}
  \caption{Some of the problems which are (anti-)metaplectically equivalent to the semi-infinite standard problem, with the same notation as in Table \ref{table: relations}. In the last row, the normalization factor of the metaplectic transformation is $C:= (t_2-t_1)/t_2$.}
  \label{table: relations2}
\end{table*}

In Table \ref{table: relations2} we present another set of quantum effects that are mathematically equivalent, not to the standard problem, but to the extended standard problem with $\alpha=\infty$, which we express, via the transformation $\sigma(x,p)=(x-\beta, p+\beta)$, as an optimization of $\Omega$ over the set of states $\S(-\infty,\beta]$.

One of these effects is a variant of quantum backflow in which the particle evolves in the presence of a constant force \cite{BMforce}. That is, with the Hamiltonian given by $H=P^2/2M -FX$. In \cite{reentrybackflow} Goussev proves that this effect is at the same time equivalent to something he calls \emph{quantum reentry}. Quantum reentry is an effect that consists in preparing a particle in $\S(-\infty,0]$, letting it evolve and then measuring a negative probability flow in some point $l \geq 0$. That is, the quantity under consideration is $-\int_{t_1}^{t_2} dt j(l,t)$ for some $t_2 > t_1 > 0$, which can again be easily transformed to the semi-infinite standard problem, as also shown in Table \ref{table: relations2}. In particular, the maximum probability transfer in both these effects is the same.

Finally, we note that the extended standard problem is equivalent to computing the maximum expression of quantum backflow when the initial momentum is in the region $[-\gamma, \infty)$ for some $\gamma \in \R$, as shown in Table \ref{table: relations2}. Thus, when the initial momentum is in this region, the probability ``backflow" acts as if there were a constant force acting on the system, since these two problems are again equivalent. This seems to have gone unnoticed by Bracken, who studied the former effect in \cite{Bracken}, despite having studied the latter in \cite{BMforce} together with Melloy.

\begin{figure*}
    \includegraphics[width=0.49\textwidth]{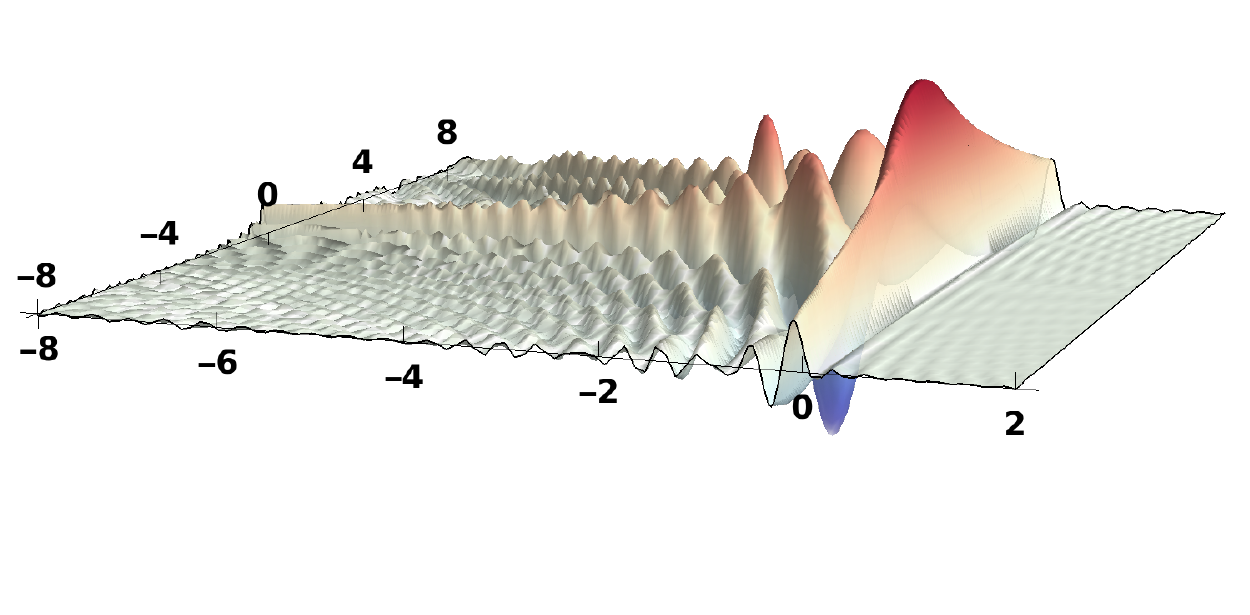}
    \includegraphics[width=0.49\textwidth]{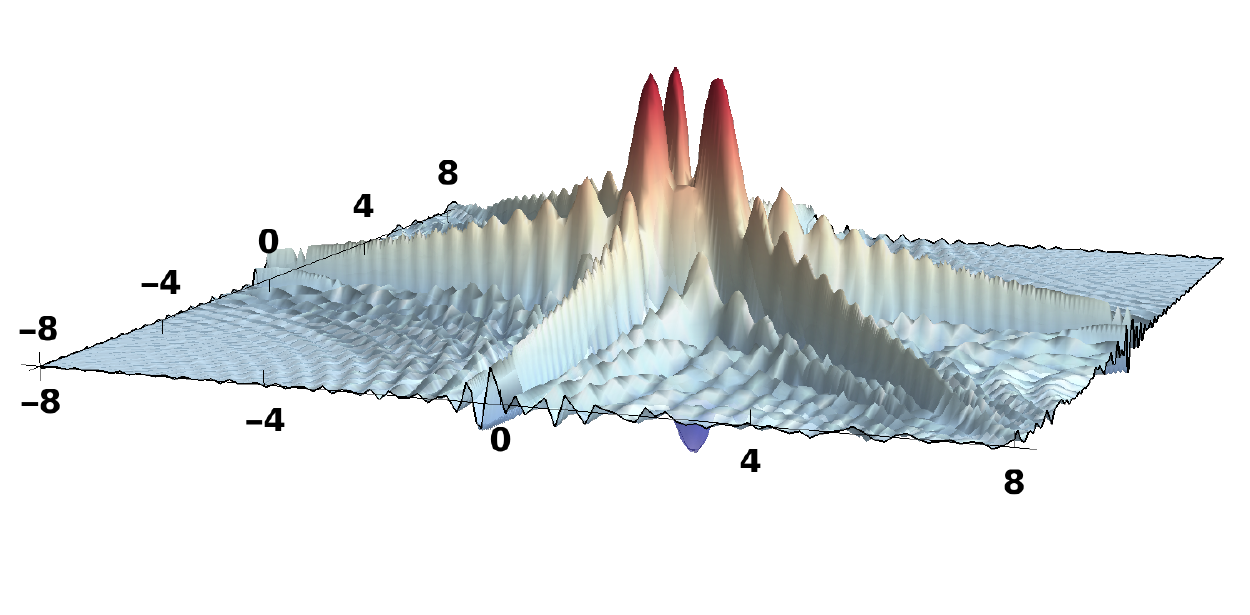}
    \caption{Wigner functions of (left) near-optimal state for the projectile scenario and (right) conjectured-optimal state for the constrained projectile scenario. Both states are obtained by truncating to the harmonic oscillator energy level $N=170$. The left state is the eigenstate of  $[\Theta(-X)]_{170}[(\Theta(X+P)-\Theta(P))]_{170}[\Theta(-X)]_{170}$ with eigenvalue $0.0331$, where $[C]_N$ denotes the restriction of the operator $C$ to the subspace spanned by the first $N+1$ number states. The right state is the eigenstate of $[\Theta(X+P)-\Theta(X)-\Theta(P)]_{170}$ with eigenvalue $0.1113$.}
    \label{fig: sameP}
\end{figure*}

\section{Classical vs. quantum rockets}
\label{rocket_sec}

The low value of $c_{bm}$ constitutes a severe obstruction to any practical application of quantum systems for transportation tasks. How to overcome this limit? A tempting idea is to consider scenarios where a transiting quantum projectile launches a second quantum projectile. Iterating this procedure, we arrive at the notion of a \emph{quantum rocket}, i.e., a quantum mechanical system that, from time to time, throws away some fuel mass in the direction opposite to the intended motion. 
Since this rocket scenario encompasses the quantum projectile scenario, its maximum quantum advantage is lower-bounded by the Bracken-Melloy constant. Furthermore, one would imagine that, should we prepare the fuel in the right quantum state, the limited quantum advantage present in quantum projectiles could be somehow bootstrapped, hence increasing the overall advantage of the quantum rocket with respect to a classical rocket whose fuel combustion has an identical momentum distribution.

Unfortunately, this is not the case, at least for a large class of quantum rockets. Consider a minimal model for a quantum rocket, where, at time $t$, the rocket itself is regarded as a $1$-dimensional particle of mass $M(t)$ and zero spin. The state of the rocket at time $t$ is therefore specified through a trace-class positive semidefinite operator $\rho(t):{\cal L}^2(\R)\to {\cal L}^2(\R)$. For most of its flight, the rocket will be propagated by the kinetic Hamiltonian $H=P_R^2/2M(t)$. At times $0=t_1<t_2<...<t_N$, though, the rocket's free evolution is interrupted: namely, at time $t_j$ the rocket burns and releases a predetermined amount of fuel $m_j$ instantaneously, thus decreasing its overall mass by the same amount.

To model the instantaneous combustion of fuel of mass $m<M$, we consider a completely positive trace-preserving (CPTP) map $\Upsilon$ that, acting on the rocket's state $\rho(t)$, returns a density matrix representing the joint state of the fuel $F$ and that of the rest of the rocket $R$, whose mass is now $M-m$, see Figure \ref{fuel_map_fig}. 

\begin{figure*}
    \includegraphics[width=0.80\textwidth]{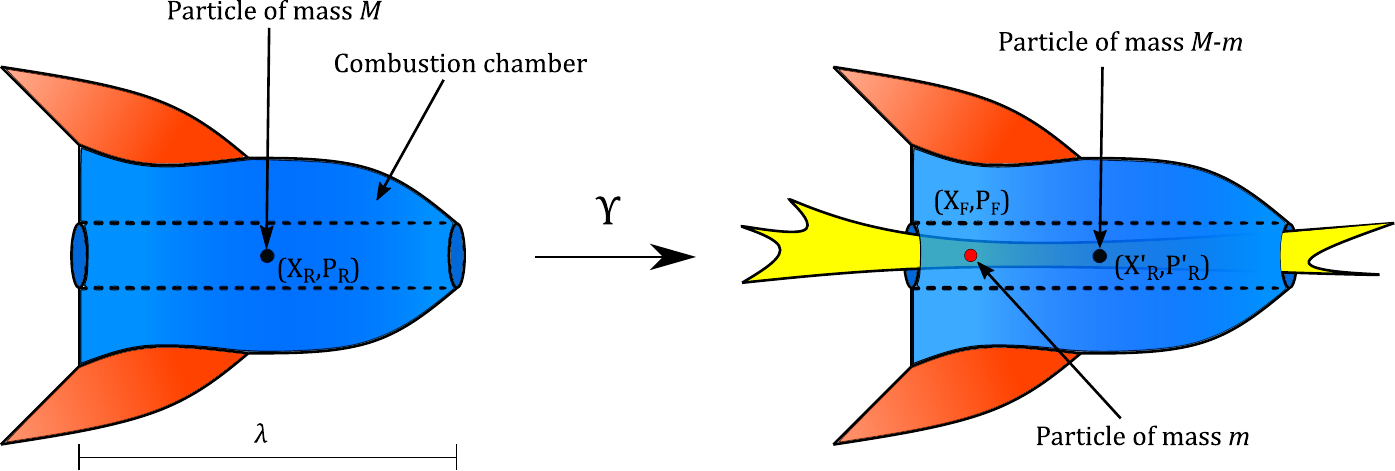}
    \caption{Action of the rocket-fuel splitting map $\Upsilon$.}
    \label{fuel_map_fig}
\end{figure*}

Call $X_F, P_F$ ($X_R, P_R$) the absolute position and momentum operators of the fuel (the rest of the rocket), and let $X_{CM},P_{CM}$ ($X_{REL},P_{REL}$) denote the canonical variables of the center of mass (the relative coordinates between systems $F$ and $R$), with:

\begin{align}
&X_{CM}=\frac{M-m}{M}X_R+\frac{m}{M}X_F, \; P_{CM}=P_R+P_F,\nonumber\\
&X_{REL}=X_F-X_R, P_{REL}=-\frac{m}{M}P_R+\frac{M-m}{M}P_F.
\label{CM_corr}
\end{align}
\noindent Let $U_{M,m}$ be the (symplectic) unitary that switches between the $R,F$ and $CM, REL$ representations and define $\omega_{CM, REL}\equiv U_{M,m}\Upsilon(\rho)U_{M,m}^\dagger$. Since $\Upsilon$ is an internal and instantaneous operation, it cannot modify the rocket's center of mass degree of freedom. This means that $\tr_{REL}(\omega)=\rho$. For $\rho=\ketbra{\psi}$, this last relation implies that $\omega=\ketbra{\psi}\otimes \sigma_{\psi}$, for some quantum state $\sigma_{\psi}$.

However, $\sigma_{\psi}$ must be independent of $\psi$. Otherwise, one could find two non-orthogonal vectors $\psi, \psi'$ with the property that $\Upsilon(\ketbra{\psi}),\Upsilon(\ketbra{\psi'})$ are more easily distinguishable than $\ketbra{\psi}, \ketbra{\psi'}$, which contradicts the contractivity of the trace norm under CPTP maps. Putting all together, we find that any rocket-fuel splitting map $\Upsilon$ must be of the form
\be
\Upsilon(\rho;\sigma,M,m)= U^\dagger_{M,m}(\rho\otimes \sigma )U_{M,m},
\label{rocket_map}
\ee
\noindent where $\sigma$ is the state of the relative system rocket-fuel. It must be noted that $\sigma$ should have been prepared in the rocket's combustion chamber. If we assume that the combustion chamber is centered in the rocket's center of mass and has length $\lambda$, then $\sigma$ must have spatial support in $[-\lambda/2,\lambda/2]$.

In describing the overall flight of the rocket, we assume that, at time $t_j$, the quantum rocket, with mass $M_j$, will release a mass $m_j$ of fuel in state $\sigma_j\in \S [-\lambda/2,\lambda/2]$ (in the relative frame of reference). Hence, the mass and state of the rocket will be instantaneously updated to $M_{j+1}= M_j-m_j$, $\rho\to \tr_F(\Upsilon(\rho;\sigma_j,M_j,m_j))$.

We consider the probability to find the rocket at time $t_{N+1}>t_N$ in the region $[a,\infty)$. This is to be compared with the maximum probability that an analog classical rocket arrives at the same region in time $t_{N+1}$. Like in the projectile scenario, this classical rocket is assumed to have, at time $t_1$, the same initial mass, initial momentum distribution and initial spatial support as the quantum one. At time $t_j$, this classical rocket will burn a mass $m_j$ of fuel, and the phase space distribution of the classical fuel in the fuel's reference frame relative to the rocket is demanded to have the same momentum distribution and spatial support as $\sigma_j$.

In these conditions, in \App E we show that the difference between the quantum and classical arrival probabilities is also limited by $c_{bm}$. This no-go result crucially relies on eq. (\ref{rocket_map}), which expresses the assumption that the fuel's interaction with the rocket is instantaneous. Physically, this corresponds to a configuration where the combustion chamber is open on both sides, i.e., the fuel is allowed to exit the rocket, not only against the rocket's direction of motion, but also towards it. Assumption eq. (\ref{rocket_map}) allows us to map the computation of the rocket's maximum quantum advantage to the standard problem (with further state constraints) through a metaplectic transformation.

\section{A transportation scenario with a quantum advantage that supersedes the Bracken-Melloy constant}
\label{variant_sec}

In view of the last result, it would be reasonable not to expect significant gaps between the arrival probabilities of quantum and classical particles. As it turns out, though, a simple variation of the way we compare classical and quantum projectiles is enough to find quantum advantages for transportation way beyond the Bracken-Melloy constant. Note that there exist known variations of the quantum backflow problem that achieve quantum advantages greater than the limit set by Bracken and Melloy \cite{new_QB, qbring, qbexperiment, qbunbounded}. Those effects are, however, unrelated to transportation tasks.

In Section \ref{class_vs_quant_proj}, we compared the behavior of a quantum projectile (or a rocket) with respect to that of a classical one with the same momentum distribution and the same spatial support at time $t=0$. Could the quantum advantage be amplified if we demanded further constraints on the initial position distribution $\mu(x)dx$ of the classical projectile, besides its support? In the extreme case, we could demand $\mu(x)dx$ to coincide with the position distribution of the quantum projectile.

Consider thus the following problem: let $\rho$ denote the density matrix of a particle of mass $M$, and let $\mu(x)dx, \nu(p)dp$ be its position and momentum distributions at time $t=0$. As before, we let the projectile evolve freely for time $\Delta T$ and then check whether the projectile is in $[a,\infty)$; call $p_q(\rho)$ the corresponding probability. How much does $p_q(\rho)$ differ from the maximum arrival probability of an analog classical particle, with initial position and momentum distributions $\mu(x)dx, \nu(p)dp$? 

The maximum classical probability of arrival is 
\begin{equation}\label{top_class}
\begin{aligned}
p^\star_c(\rho) = & \sup \int dxdp W(x,p)\Theta\left(x+p\frac{\Delta T}{M}-a\right)\\
&\text{s.t. } \quad \forall x,p, W(x,p)\geq0,\\
&\phantom{s.t. } \quad \int dpW(x,p) = \mu(x),\\
&\phantom{s.t. } \quad \int dxW(x,p) = \nu(p),\\
\end{aligned}
\end{equation}
where $W(x,p)$ represents the probability distribution of the classical particle in phase space at time $t=0$. 

The maximum quantum-to-classical advantage in this projectile scenario is therefore $\Phi^\star=\sup_{\rho\in S}\mathbb{W}(\rho)$, where $\mathbb{W}(\rho):=p_q(\rho)-p^\star_c(\rho)$. This is a nested max-min optimization problem, whose solution can be proven independent of $a, M,\Delta T$ \footnote{This can be shown by replacing $\rho$ by $U \rho U^\dagger$ in $\mathbb{W}(\rho)$, where $U$ is the metaplectic transformation $X\to\sqrt{\frac{\Delta T}{M}}X+a$, $P\to\sqrt{\frac{M}{\Delta T}}P$.}.

In \App F, $p_c^\star(\rho)$ is shown to equal $s(\infty)$, the solution of the system of ordinary differential equations
\begin{align}
    \frac{ds}{dx}=&\Theta_{+}(q)\mu(x)+(1-\Theta_{+}(q))\min\left(\mu(x),\tilde{\nu}(a-x)\right),\nonumber\\
    \frac{dq}{dx}=&\Theta_{+}(q)(\tilde{\nu}(a-x)-\mu(x))+\nonumber\\
    &(1-\Theta_{+}(q))\max\left(\tilde{\nu}(a-x)-\mu(x),0\right),
    \label{ODE_main}
\end{align}
\noindent with initial conditions $s(-\infty)=q(-\infty)=0$. Here $\Theta_{+}(z)$ is meant to be $1$ for $z>0$ and $0$ otherwise. $s(\infty)$ can be computed numerically via, e.g., Euler's explicit method.

Since we know how to compute $p^\star_c(\rho)$, one could, in principle, use gradient ascent methods to find the maximum of $\mathbb{W}(\rho)$, over all quantum states with, say, support on the space spanned by the first $N$ number basis vectors. That is, we could parametrize any such state $\rho$ as $\rho=\sum_{m,n=0}^N\rho_{m,n}\ket{m}\bra{n}$ and then follow the gradient of $\mathbb{W}(\rho)$ with respect to the variables $\rho_{m,n}$. Unfortunately, $\mathbb{W}$ is a concave function, so the method is not guaranteed to converge to the absolute maximum. Moreover, we empirically observe that, starting from a random state, projected gradient methods typically converge to very suboptimal values.

To find a suitable starting point for gradient ascent, we considered the following approach: suppose that there existed a linear operator $Z$ such that 
\begin{equation}
p^\star_c(\rho)\leq \tr(Z\rho),    
\label{special_operator}
\end{equation}
\noindent for all states $\rho$. Then we could maximize the value 
\begin{equation}
\mathbb{W}_Z(\rho):=\tr\left[\rho\left(\Theta\left(X+P\frac{\Delta T}{M}-a\right)-Z\right)\right]
\end{equation}
\noindent over all density matrices with support on the first $N$ number states. The result would provide us with a lower bound on $\Phi^\star$. In addition, if the maximizer $\rho^\star$ satisfied $\mathbb{W}_Z(\rho^\star)>0$, then that state would be a good starting point for gradient ascent.

Now, how to identify an operator $Z$ satisfying (\ref{special_operator})? Let $f, g:\R\to\R$ be two functions such that 
\begin{equation}
\Theta(x+p\Delta T/M-a)-f(x)-g(p) \leq 0,
\label{dual_condi}
\end{equation}
\noindent for all $x,p$. Then, for any distribution $W(x,p)$ in phase space with marginals $\mu(x),\nu(p)$,
\begin{align}
&\int dxdp W(x,p)\Theta(x+p\Delta T/M-a)\leq \nonumber\\
&\int dxdp W(x,p)(f(x)+g(p))=\nonumber\\
&\int dx\mu(x)f(x)+\int dp\nu(p)g(p).
\label{dual_idea}
\end{align}
\noindent It follows that the operator $Z=f(X)+g(P)$ fulfills condition (\ref{special_operator}). In fact, the dual of problem \eqref{top_class} is the maximum of the right-hand side of eq. (\ref{dual_idea}) over all such functions $f,g$.

Take $M=\Delta T=1, a=0$. We observe that the functions $f=g=\Theta$ satisfy (\ref{dual_condi}), and hence, the supremum of the spectrum of the operator $\Omega = \Theta(X+P)-\Theta(P)-\Theta(X)$ provides us with a lower bound for $\Phi^\star$, as $\tr(\rho\Omega) \geq \Phi^*$.

If we truncate this operator in the number basis, we are looking at the maximum eigenvalue of the matrix $({\cal M}^{(N)}_{nm}:n,m=0,...,N)$, with
\[
{\cal M}^{(N)}_{nm}=\bra{n}\left(\Theta(X+P)-\Theta(X)-\Theta(P)\right)\ket{m},
\]
For $N=170$, the maximum eigenvalue of this matrix is $0.1113$: the reader can find a plot of the Wigner function of the corresponding eigenvector in Figure \ref{fig: sameP} (right). Taking $N=1700$, we obtain the tighter bound $\Phi^\star\geq 0.1228$. The maximum quantum advantage in this projectile scenario is therefore substantially greater than the conjectured value of $c_{bm}$, or even its upper bound $0.0725$, derived in section \ref{solving_standard_sec}.

Applying gradient methods on those states to improve their $\mathbb{W}$ value proved to be tricky, though. Call $\rho^\star$ the state  corresponding to the eigenvector of ${\cal M}^{(N)}$. We observe that, even for low values of $N$ (say, $N=30$), we need to use a very small step size in eq. (\ref{ODE_main}) to estimate $p^\star_c(\rho^\star)$ precisely. When we do so, we find that $p^\star_c(\rho^\star)\approx \tr\{\rho^\star(\Theta(X)+\Theta(P))\}$: that is, for such quantum states, our upper bound (\ref{special_operator}) on $p^\star_c$ is (approximately) tight. Around the eigenvectors of ${\cal M}^{(N)}$, the gradient of $\mathbb{W}$ explodes, possibly because the function is not everywhere differentiable. Using random perturbations of $\rho^\star$ as a seed, projected gradient methods only produced states with a objective value slightly smaller than $\mathbb{W}(\rho^\star)$.

From all the above, it is thus natural to conjecture that the obtained value of $0.1228$ is (close to) a local maximum of $\mathbb{W}$, at least among quantum states with support in $\{\ket{n}:n=0,...,1700\}$. 

On the other hand, note that after a suitable metaplectic transformation the problem $\sup_{\rho}\tr(\rho \Omega)$ becomes $\sup_{\rho}\tr(\rho  \tilde{\Omega})$, where
\[
\tilde{\Omega} = \id - \sum_{k=0}^2 \Theta (X_k) = -\frac{1}{2} \id - \frac{3}{2}\left(\frac{1}{3} \sum_{k=0}^2 \text{sgn}(X_k) \right)
\]
with $X_k := \text{cos}(2\pi k /3) X + \text{sin}(2\pi k /3) P$. The operator $\sum_{k=0}^2 \text{sgn}(X_k)/3 $ is the one studied by Tsirelson in \cite{tsirelson}. The best known bounds for its spectrum are given in \cite{valerio}. Using Equation (D20) in \cite{valerio}, one obtains that $\Phi^* \geq -0.5 + 1.5 \times \sqrt{0.17491} = 0.1262$. In particular, this shows how unreliable the numerical estimation of these quantities is, even after using a basis with $1700$ number states, and thus the importance of getting good upper bounds as well as lower bounds.



\section{Conclusion}
\label{conclusion_sec}

In this letter, we have investigated how the dynamics of quantum and classical projectiles differ, using the probability of arrival at a distant region of space as a figure of merit. We found that non-relativistic quantum particles can arrive at a distant region with higher or lower probability than any classical particle with the same initial spatial support and momentum distribution. Curiously enough, the maximum gap between quantum and classical probabilities is independent of the distance to the arrival region, and just depends on the mass $M$ and spatial support $L$ of the projectile and its flying time $\Delta T$ through the single parameter $\alpha=ML^2/\Delta T$. 

The discrepancy between the quantum and classical arrival probabilities is, however, limited by the Bracken-Melloy constant $c_{bm} \approx 0.0384517$. As we showed, the maximum quantum advantage of rockets with an open combustion chamber is also bounded by this value. Our no-go result does not apply, however, to rockets with a $1$-side closed combustion chamber, which just allows the fuel to exit the rocket opposite to its direction of motion. Whether such rocket models are also limited by $c_{bm}$, or on the contrary, they can achieve arrival probabilities much higher than classical is an interesting topic for future research.

In a similar direction, we showed that considerable quantum-classical gaps of at least $0.1262$ can be observed if we demand classical projectiles to reproduce the initial position distribution of the quantum projectile. It is an open problem whether this figure is indeed close to the maximum quantum advantage, and whether this effect can be exploited for real transportation tasks.

\vspace{10pt}
\noindent\emph{Acknowledgements}

We thank Valerio Scarani for interesting discussions, Reinhard Werner for pointing us to reference \cite{Werner_1988}, and Zaw Lin Htoo for pointing to us that using (D20) of \cite{valerio} one can slightly improve our numerical bound of Section V. D.T. is a recipient of a DOC Fellowship of the Austrian Academy of Sciences at the Institute of Quantum Optics and Quantum Information (IQOQI), Vienna. T.P.L is supported by the Lise Meitner Fellowship of the Austrian Academy of Sciences (project number M 2812-N).

\vspace{10pt}
\noindent\emph{Competing Interests}
The Authors declare no Competing Financial or Non-Financial Interests.

\vspace{10pt}
\noindent\emph{Author Contributions}
All authors contributed equally to the research and writing.

\vspace{10pt}
\noindent\emph{Data availability}
There is no data set in this work.

\vspace{10pt}
\noindent\emph{Disclaimer}
This version of the article has been accepted for publication, after peer review but is not the Version of Record and does not reflect post-acceptance improvements, or any corrections. The Version of Record is available online at: https://doi.org/10.1038/s41534-023-00739-z.

\bibliography{biblios2}

\newpage

\section{Appendix}

\begin{appendix}

In this Appendix we perform some of the lengthy computations that give the claims of the main text. It is organized as follows. In section A, we review the relevant properties of the Wigner function. In section B we give an explicit formula for $\varphi(\alpha)$ and show how to numerically approximate it for finite $\alpha$. In section C we compute a numerical lower bound for $\varphi(\infty)$. In section D we compute a numerical upper bound for $\varphi(\infty)$. Section E proves that there is no advantage in our model of the quantum rocket with respect to a single projectile. Finally, in section F we describe the numerical methods used in the restricted projectile scenario.

\section{Notes on the Wigner function}
\label{wigner_app}

The Wigner function of a quantum state $\rho$ is
\[
W_\rho (x,p):=\frac{1}{2\pi}\int_{-\infty}^{\infty} dy \bra{x-\frac{y}{2}}\rho \ket{x+\frac{y}{2}} e^{ipy}.
\]
This is a partial Fourier transform on the function $\bra{x-y/2}\rho \ket{x+y/2}$, and as such is defined with the usual density arguments from the states $\rho$ such that $\bra{x}\rho \ket{y}$ is a Schwarz function of two variables.

We now study the action of metaplectic transformations on the Wigner function. Suppose then that $\sigma:\R^2 \rightarrow \R^2$ is affine-linear, and $[\sigma(X,P)_1,\sigma(X,P)_2]=[X,P]$. Then, calling $\tilde{\sigma}$ the linear part of $\sigma$, we must conclude that $\tilde{\sigma} \in SL_2(\R)$. It is well known that $SL_2(\R) = Sp_2(\R)$, which gives rise to the name \emph{metaplectic} that we have used in the main text. Furthermore, the KAN decomposition of $SL_2(\R)$ is
\[
SL_2(\R) = SO_2(\R) \cdot \begin{pmatrix} \mu & 0 \\ 0 & 1/\mu \end{pmatrix} \cdot  \begin{pmatrix} 1 & \nu \\ 0 & 1 \end{pmatrix}.
\]
That is, every matrix decomposes as a product of a rotation, a dilation and a translation. These all correspond to time evolutions of quadratic Hamiltonians ($P^2+X^2$, $XP+PX$ and $P^2$ or $X^2$, respectively). Finally, the affine part of the map can be realized by time-evolving with the Hamiltonians $X$ and $P$. These six Hamiltonians thus give rise to the unitaries $U_\sigma$ mentioned in the main text.

On the other hand, since all such Hamiltonians are at most quadratic in momentum and position, the time evolution of the Wigner function must satisfy the Liouville equation \cite{quadratic1, quadratic2}, i.e., it must evolve classically in phase space. Therefore,
\[
W_\rho(\sigma^{-1}(x,p)) = W_{U_\sigma \rho U_\sigma^\dagger}(x,p),
\]
as the main text claims.

If, on the other hand, $[\sigma(X,P)_1, \sigma(X,P)_2]=-[X,P]$, then there is an extra matrix 
\[
\begin{pmatrix} 1 & 0 \\ 0 & -1 \end{pmatrix}
\]
in the KAN decomposition of the linear part of $\sigma$. This operator corresponds to the antiunitary map $\rho \mapsto \rho^*$, where $*$ denotes complex conjugation in the position basis. Indeed, a short computation shows that
\[
W_{\rho^*}(x,p) = W_{\rho}(x,-p).
\]
Given an operator $\Omega$, we define its Wigner function as 
\[
W_\Omega(x,p) := \int_{-\infty}^{\infty} dy \bra{x-\frac{y}{2}}\Omega \ket{x+\frac{y}{2}} e^{ipy}.
\]
With these choices of normalization, a short computation shows that
\[
\tr(\rho \Omega) = \int_{\R^2} dx dp W_\rho(x,p) W_\Omega(x,p).
\]

In the main text we are primarily concerned with operators of the form $f(aX+bP+c)$ for some bounded measurable function $f$. We now prove that
\begin{equation} \label{eq_wigner}
\tr(\rho f(aX+bP+c)) = \int_{\R^2} dx dp f(ax+bp+c) W_{\rho}(x,p).
\end{equation}

Since $f$ is bounded and measurable, we have that (as a tempered distribution) it has an inverse Fourier transform $\hat{f}$, and we may write $f(x)=\int_\R dt \hat{f}(t)e^{ixt}$. Via functional calculus, we thus have
\begin{align*}
f(aX+bP+c) &= \int_\R dt \hat{f}(t) e^{it(aX+bP+c)}\\ 
& = \int_\R dt \hat{f}(t)e^{itbP}e^{itaX}e^{-it^2ab/2}e^{itc},
\end{align*}
where we have used the Baker-Campbell-Haussdorf formula $e^{i(\xi X+\zeta P)}=e^{i\zeta P}e^{i\xi X}e^{-i\frac{\xi \zeta}{2}}$. A straightforward computation now shows that
\[
\bra{x-y/2}f(aX+bP+c)\ket{x+y/2} = \int_\R dt \hat{f}(t) e^{it(ax+bp+c)},
\]
from which we conclude the result.

Since it will be useful soon, we next compute the Wigner function of the operator $\ketbra{m}{n}$, in number basis, i.e., $(X+iP)\ket{n}=\sqrt{2n}\ket{n-1}$. First note that, by linearity of the Wigner function, for any state $\rho=\sum_{m,n}\rho_{mn}\ketbra{m}{n}$, we have that
\begin{equation}\label{eq:Wmn}
    W_\rho(x,p)=\sum\rho_{mn}W_{\ketbra{m}{n}}(x,p).
\end{equation}

Now, take $\rho$ to be a coherent state, i.e., $\rho=\ketbra{\alpha}{\alpha}$, with
\begin{equation}
\ket{\alpha}=e^{-|\alpha|^2/2}\sum_{k=0}^\infty\frac{\alpha^k}{\sqrt{k!}}\ket{k}.
\end{equation}
It follows that
\begin{equation}
    \rho_{mn}=e^{-\abs{\alpha}^2}\frac{\alpha^m\bar\alpha^n}{\sqrt{m!n!}}.
\end{equation}
On the other hand, the Wigner function of a coherent state is known to be \cite{gaussian}
\begin{equation}\label{eq:Wcoh}
W_{\rho}(x,p)=\frac{1}{\pi}e^{-r^2-2|\alpha|^2+\sqrt{2}(\alpha(x-ip)+\bar\alpha(x+ip))},
\end{equation}
with $r^2=x^2+p^2$. Cancelling the factor $e^{-\abs{\alpha}^2}$ in both sides of \eqref{eq:Wmn} and expanding the remaining exponential in \eqref{eq:Wcoh} as a power series in $\alpha,\bar\alpha$, we can compare the coefficients multiplying $\alpha^m\bar\alpha^n$ on both sides of the resulting equation, thus obtaining

\begin{align}\label{wigner_mn}
&W_{\ket{m}\bra{n}}(x,p)= \\
&\frac{\sqrt{m!n!}}{\pi}e^{-r^2}\sum_{k=0}^{\min(m,n)}\frac{(-1)^k}{k!}\frac{(\sqrt{2}r)^{m+n-2k}}{(m-k)!(n-k)!}e^{i\theta(n-m)}, \nonumber \\
&\theta=\arg(x+ip). \nonumber
\end{align}

Note that in \cite{tsirelson} Tsirelson provides the complex conjugated formula for the same quantity. This mistake does not, however, invalidate the main result of \cite{tsirelson}, namely, the computation of the spectrum of a given linear operator. This is so because the spectra of a self-adjoint operator and its complex conjugate in a given basis coincide.

Next, we invoke \eqref{wigner_mn} to derive the matrix elements $\mathcal{O}_{nm}(\phi):=\bra{n}\Theta(\cos(\phi)X+\sin(\phi)P)\ket{m}$ and show that
\begin{widetext}
\begin{align}\label{theta_exp}
&\mathcal{O}_{nm}(\phi)=\frac{\sqrt{m!n!}}{\pi}\frac{e^{i\phi(n-m)}(i^{n-m}-i^{m-n})}{i(n-m)}\sum_{k=\max(m,n)}^{m+n}\frac{(-1)^{m+n-k}2^{k-\frac{m+n}{2}-1}\Gamma\left(k-\frac{m+n}{2}+1\right)}{(m+n-k)!(k-n)!(k-m)!}.
\end{align}
\end{widetext}
\noindent We will use this expression in Appendices \ref{low_bm_app}, \ref{variant_app} to lower bound the maximum quantum advantage in the standard and restricted projectile scenarios.

To begin, from eq. (\ref{eq_wigner}) we have that
\begin{equation}
    \mathcal{O}_{nm}(\phi) = \int dx dp W_{\ketbra{m}{n}}(x,p)\Theta(x\cos\phi+p\sin\phi).
\end{equation}
\noindent We can evaluate the right-hand side of the above equation by changing to polar coordinates. The result is
\begin{equation}
{\cal O}_{nm}(\phi)=\frac{\sqrt{m!n!}}{\pi}w_{nm}\frac{e^{i\phi(n-m)}(i^{n-m}-i^{m-n})}{i(n-m)},
\label{theta_interm}
\end{equation}
\noindent with
\begin{align}
w_{nm}&:=\sum_{k=0}^{\min(m,n)}\frac{(-1)^k2^{\frac{m+n}{2}-k-1}\Gamma\left(\frac{m+n}{2}-k+1\right)}{k!(m-k)!(n-k)!}\nonumber\\
&=\sum_{k=\max(m,n)}^{m+n}\frac{(-1)^{m+n-k}2^{k-\frac{m+n}{2}-1}\Gamma\left(k-\frac{m+n}{2}+1\right)}{(m+n-k)!(k-n)!(k-m)!},
\end{align}
\noindent where, in the last step, we changed the sum variable $k\to m+n-k$ so that a comparison with eq. (1.5) in \cite{tsirelson} can be made. 

As it turns out, the final expression for $w_{nm}$ can be written in terms of the generalized hypergeometric function $\phantom{}_pF_q$. Thanks to such an identity, we were able to compute $w_{nm}$ accurately for large values of $m,n$.

\section{Properties of $\varphi(\alpha)$}
\label{varphi_app}
First, we will calculate the kernel of $\Omega$ in position representation. Our starting point is the identity
\begin{equation}
\mbox{sign}(A)=\frac{1}{i\pi}\int \frac{dt}{t}e^{itA},    
\end{equation}
\noindent where the integral must be understood as a Cauchy principal value. Thus we have that
\begin{align}
&\mbox{sign}(P+X)=\frac{1}{i\pi}\int \frac{dt}{t}e^{it(P+X)}\nonumber\\
&=\frac{1}{i\pi}\int \frac{dt}{t}e^{itP}e^{it X}e^{-i\frac{t^2}{2}}\nonumber\\
&=\frac{1}{i\pi}\int \frac{dt}{t}dxdydp \ket{x}\braket{x}{p}e^{itp}\braket{p}{y}\bra{y}e^{it y}e^{-i\frac{ t^2}{2}}\nonumber\\
&=\frac{1}{i\pi}\int \frac{dt}{t}dxdydp \frac{e^{ip(t-y+x)}}{2\pi}e^{it y}e^{-i\frac{t^2}{2}}\ket{x}\bra{y}\nonumber\\
&=\frac{1}{i\pi}\int \frac{dt}{t}dxdy \delta(t-y+x)e^{it y}e^{-i\frac{t^2}{2}}\ket{x}\bra{y}\nonumber\\
&=\frac{1}{i\pi}\int dxdy \frac{e^{\frac{i}{2} (y^2-x^2)}}{y-x}\ket{x}\bra{y}.
\end{align}
\noindent To arrive at the final expression, we invoked the Baker-Campbell-Haussdorf formula $e^{i(\xi X+\zeta P)}=e^{i\zeta P}e^{i\xi X}e^{-i\frac{\xi \zeta}{2}}$ in the first line; the resolution of the identity $\id=\int dx \ketbra{x}=\int dx \ketbra{y}$, in the second one; the relation $\braket{x}{p}=e^{ipx}/\sqrt{2\pi}$ (assuming that the bra is an element of the position basis; and the ket, of momentum basis), in the third one; and the relation $\int dp e^{ips}=2\pi\delta (s)$, in the fourth one.
    
Using the same techniques, one finds that
\be
\mbox{sign}(P)=\frac{1}{i\pi}\int dxdy \frac{1}{y-x}\ketbra{x}{y}.
\ee
Hence we have that
\begin{align}
    &\Omega=\frac{1}{2}\left(\mbox{sign}(X+P)-\mbox{sign}(P)\right)=\nonumber\\
    &\frac{1}{2\pi i}\int dxdy \frac{e^{\frac{i}{2} (y^2-x^2)}-1}{y-x}\ketbra{x}{y}.
\end{align}
This expression can be further reduced to a real kernel by conjugating it with the unitary $e^{\frac{i}{4}X^2}$, which results in the operator
\be
\tilde{\Omega}=\frac{1}{4\pi}\int dxdy (x+y)\mbox{sinc}\left(\frac{1}{4} (y^2-x^2)\right)\ketbra{x}{y},
\label{tilde_Omega}
\ee
where $\mbox{sinc}(z):= sin(z)/z$.

\subsubsection{Bounding $\varphi(\alpha)$}

We start from the easily verifiable identity:
\be
\mbox{sinc}(y)=\frac{1}{2}\int_{-1}^1e^{i\omega y}d\omega.
\label{id_sinc2}
\ee
Applying the identity to eq. (\ref{tilde_Omega}), we find that
\be
\tilde{\Omega}=\frac{1}{2}\int_{-1}^1 d\omega A_\omega,
\label{decomp}
\ee
with
\be
A_\omega=\frac{1}{4\pi}\int dxdy\ketbra{x}{y}(x+y)e^{i\omega\frac{y^2-x^2}{4}}.
\ee
\noindent Define $S(\alpha):={\cal S}([-\sqrt{\alpha},0])$. Note that $A_\omega=U_\omega A_0 U_{\omega}^\dagger$, where the unitary $U_\omega=e^{-\frac{i}{4}\omega X^2}$ leaves ${\cal S}(\alpha)$ invariant. Now, by eq. (\ref{decomp}), we have that
\begin{align}
\sup_{\rho\in S(\alpha)} \tr\left(\rho\tilde{\Omega}\right) &\leq \frac{1}{2}\int_{-1}^1 d\omega \sup_{\rho\in S(\alpha)}\tr\left(\rho A_\omega\right)\nonumber\\
&= \sup_{\rho\in S(\alpha)}\tr\left(\rho A_0\right).
\end{align}

On the other hand, when averaged over elements of $S(\alpha)$, $A_0$ has support on a two-dimensional subspace, namely, the span of the vectors
\be
\ket{\psi_0}=\int_{[-\sqrt{\alpha},0]}dx\ket{x},\ket{\psi_1}=\int_{[-\sqrt{\alpha},0]}xdx\ket{x}.
\ee
The maximum eigenvalue of $A_0$ is therefore the result of solving the generalized eigenvalue problem $\min\{\lambda:\lambda G-F\geq0\}$ with $2\times 2$ matrices $F,G$ given by
\be
F_{jk}=\braket{\psi_j}{\psi_0}\braket{\psi_1}{\psi_k}+\braket{\psi_j}{\psi_1}\braket{\psi_0}{\psi_k}, G_{jk}=\braket{\psi_j}{\psi_k}.
\ee
The result is the upper bound on $\varphi(\alpha)$
\be
\varphi(\alpha)\leq\frac{(2\sqrt{3}-3)}{24\pi}\alpha.
\label{linear_upper}
\ee
As shown in Figure 2 (main text), this analytic (and linear) bound is very good for small values of $\alpha$. 

To arrive at better approximations for $\varphi(\alpha)$, we will exploit the fact that the integrand in (\ref{tilde_Omega}) is an analytic function; and the spatial support of the states in $S(\alpha)$, finite. Consider an operator of the form
\be
O:= \int_{[-\sqrt{\alpha},0]^2}dxdzf(x,z)\ketbra{x}{z},
\ee
with $f(x,z)$ analytic in $[-\sqrt{\alpha},0]^2$. Then, for any $\epsilon>0$, one can find $N$ such that the $N^{th}$-order Taylor expansion $f_N(x,z)=\sum_{n,m=0}^Nf_{m,n}x^mz^n$ of $f(x,z)$ satisfies $|f_N(x,z)-f(x,z)|<\epsilon$, for $x,z\in [-\sqrt{\alpha},0]$. Define thus the operator
\be
O_N:=\int_{[-\sqrt{\alpha},0]^2}dxdzf_N(x,z)\ketbra{x}{z},
\ee
and let $\ketbra{\psi}\in S(\alpha)$. Then,
\begin{align}
&|\bra{\psi}(O-O_N)\ket{\psi}|\nonumber\\
&\leq\int_{[\sqrt{\alpha},0]^2}dxdz|f(x,z)-f_N(x,z)||\psi(x)||\psi(y)| \nonumber\\
&\leq\epsilon\int_{[-\sqrt{\alpha},0]^2}|\psi(x)||\psi(y)|=\epsilon|\braket{|\psi|}{\psi_0}|^2\leq\epsilon\sqrt{\alpha},
\end{align}
where $\ket{|\psi|}$ denotes the normalized state with wave-function $|\psi(x)|$. The maximum eigenvalue of $O_N$ is therefore an $\epsilon\sqrt{\alpha}$-approximation to the top of the spectrum of $O$. Note, as we did in deriving the upper bound on $\phi(\alpha)$, that $O_N$ has support on the finite set of vectors
\be
\Psi_N:=\left\{\int_{[-\sqrt{\alpha},0]} dxx^k\ket{x}, k=0,...,N\right\},
\ee
hence in principle we can diagonalize it exactly. In practice, though, the Gram matrix of the vectors in $\Psi_N$ is ill-conditioned, so, with computer precision, one can just diagonalize $O_N$ reliably for low values of $N$. To overcome this difficulty, we exploit the properties of the Legendre polynomials \cite{abramowitz+stegun}.

The Legendre polynomials $P_n(x)$ are o
rthogonal in the interval $[-1,1]$ with respect to the weight $w(x)=1$. It follows that $P_n(1+2x/\sqrt{\alpha})$ are orthogonal in the interval $[-\sqrt{\alpha},0]$. Invoking the formula for the scalar product of Legendre polynomials \cite{abramowitz+stegun}, and taking into account the compression $[-1,1]\to [-\sqrt{\alpha},0]$, we have that
\be
\int_{[-\sqrt{\alpha},0]} dx P_n\left(1+\frac{2x}{\sqrt{\alpha}}\right)P_m\left(1+\frac{2x}{\sqrt{\alpha}}\right)=\frac{\sqrt{\alpha}}{2n+1}\delta_{m,n}.
\ee
In addition, Legendre polynomials satisfy the recurrence relation \cite{abramowitz+stegun}
\begin{align}
xP_n\left(1+\frac{2x}{\sqrt{\alpha}}\right)=&\frac{\sqrt{\alpha}(n+1)}{2(2n+1)}P_{n+1}\left(1+\frac{2x}{\sqrt{\alpha}}\right)\nonumber\\
&-\frac{\sqrt{\alpha}}{2} P_{n}\left(1+\frac{2x}{\sqrt{\alpha}}\right)+\nonumber\\
&+\frac{\sqrt{\alpha}n}{2(2n+1)}P_{n-1}\left(1+\frac{2x}{\sqrt{\alpha}}\right).
\end{align}
From the identity $P_0\left(1+2x/\sqrt{\alpha}\right)=1$, we hence arrive at a simple expression for the expansion of the monomials $\{x^k:k=0,...,N\}$ in the polynomial basis ${\cal B}_N=\{P_n\left(1+2x/\sqrt{\alpha}\right):n=0,...,N\}$:
\be
\sum_{n=0}^N\bra{n}\hat{X}^k\ket{0}P_n\left(1+\frac{2x}{\sqrt{\alpha}}\right),
\ee
where $\hat{X}$ is the $N+1\times N+1$ matrix with rows and columns numbered from $0$ to $N$ and non-zero coefficients
\begin{align}
\hat{X}_{n+1,n}=\frac{\sqrt{\alpha}(n+1)}{2(2n+1)},\hat{X}_{n,n}=-\frac{\sqrt{\alpha}}{2},\hat{X}_{n-1,n}=\frac{\sqrt{\alpha}n}{2(2n+1)}.
\end{align}
Call $\hat{G}$ the Gram matrix of the basis functions ${\cal B}_N$, i.e., $\hat{G}_{mn}=\frac{\sqrt{\alpha}}{2n+1}\delta_{m,n}$. From the above it follows that 
\begin{align}
\hat{F}_{mn}&:=\int_{[0,1]^2}dxdyf_N(x,y)P_m\left(1+\frac{2x}{\sqrt{\alpha}}\right)P_n\left(1+\frac{2x}{\sqrt{\alpha}}\right)\nonumber\\
&=\bra{m}\hat{G}\left(\sum_{j,k=0}^Nf_{jk}\hat{X}^j\ketbra{0}{0}(\hat{X}^k)^T\right)\hat{G}\ket{n}.
\end{align}

Diagonalizing $O_N$ thus entails solving the generalized eigenvalue problem
\begin{align}
\max &\bra{\Psi}\hat{F}\ket{\Psi}\nonumber\\
\mbox{s.t. }&\bra{\Psi}\hat{G}\ket{\Psi}=1.
\end{align}
Defining $\ket*{\hat{\Psi}}=\hat{G}^{1/2}\ket{\Psi}$, we find that our $\epsilon$-approximation to the maximum eigenvalue of $O$ is the maximum eigenvalue of the matrix $\hat{G}^{-1/2}\hat{F}\hat{G}^{-1/2}$.

Let us apply these considerations to the operator (\ref{tilde_Omega}). In this case, $f(x,z)=\frac{1}{4\pi}(x+y)\mbox{sinc}\left(\frac{(y^2-x^2)}{4}\right)$. By Taylor's remainder theorem, we have that, for any $z$, 
 \be
 \sin(z)=\sum_{k=0}^N\frac{(-1)^{k}}{(2k+1)!}z^{2k+1}+\frac{\sin^{(2N+3)}(\xi)}{(2N+2)!}\xi^{2N+2}z,
 \ee
for some $\xi\in [0,z]$. It follows that, for $x,y\in [-\sqrt{\alpha},0]$, the following relation holds:
\begin{align}
&\left| \frac{1}{4\pi} (x+y)\mbox{sinc}\left(\frac{y^2-x^2}{4}\right)\right.\nonumber\\
&\left. -\frac{1}{4\pi}(x+y)\sum_{k=0}^N\frac{(-1)^k}{(2k+1)!}\left(\frac{y^2-x^2}{4}\right)^{2k}\right|\nonumber\\
&\leq \frac{\sqrt{\alpha}}{2\pi}\frac{1}{(2N+2)!}\left(\frac{\alpha}{4}\right)^{2N+2}=\epsilon.
\end{align}
Using Stirling's approximation $\ln(n!)\approx n\ln(n)-n$, we conclude that, in order to compute $\varphi(\alpha)$ up to error $\delta\equiv \epsilon\sqrt{\alpha}$, the following condition must be fulfilled:
\be
(2N+2)(\ln(2N+2)-1-\ln(\alpha))\gtrapprox O\left(\ln\left(\frac{1}{\delta}\right)\right)+O(\ln(\alpha)).
\ee
A sufficient condition to satisfy this relation is that $\ln(2N+2)-1-\ln(\alpha)\geq 1$ and $(2N+2)\geq \ln\left(1/\epsilon\right)+O(\ln(\alpha))$.

\section{Lower bounds on $\varphi(\infty)=c_{bm}$}
\label{low_bm_app}
To improve the lower bounds on $c_{bm}$ obtained through the exact computation of $\varphi(\alpha)$ for high values of $\alpha$, we will follow a variational approach. Recall that $c_{bm}$ is the result of maximizing $\tr(\Omega\rho)$ over all quantum states $\rho\in {\cal S}(-\infty,0]$. Hence, any quantum state satisfying this constraint gives a lower bound on $c_{bm}$. For any $\rho$, we can enforce this constraint by projection:
\begin{equation}
\hat{\rho}:=\frac{\Theta(-X)\rho \Theta(-X)}{1-\epsilon} \in {\cal S}(-\infty,0]
\end{equation}
with $\epsilon = 1 - \tr(\Theta(-X)\rho)$. Using $\|\Omega\|_\infty\leq 1$, it is easy to prove that, for $\rho=\ketbra{\psi}{\psi}$,
\begin{equation}
\tr(\hat{\rho}\Omega)\geq \frac{\tr(\rho\Omega)}{1-\epsilon}-2\sqrt{\frac{\epsilon}{1-\epsilon}}-\frac{\epsilon}{1-\epsilon}.
\label{epsilon_corr}
\end{equation}
which provides a way to lower bound $c_{bm}$, given an arbitrary quantum state not necessarily in ${\cal S}(-\infty,0]$.

Consider, thus, a state $\rho$ with support in $\H_N=\mbox{span}\{\ket{n}:n=0,\dots,N\}$. The restrictions of the operators $\Omega=\Theta(X+P)-\Theta(P)$ and $\Theta(-X)$ to $\H_N$ can be computed through eq. (\ref{theta_exp}). Taking $N=1000$, we find, via matrix diagonalization, the pure state $\ket{\psi}\in\H_N$ maximizing the overlap
\begin{equation}
\bra{\psi}\left(\Omega_N+\lambda\Theta(-X)_N\right)\ket{\psi},
\end{equation}
\noindent with $\lambda=2500$. Defining $\rho^\star:=\ketbra{\psi}{\psi}$, we compute the averages $\tr(\rho\hat{\Omega})$, $\tr(\rho\Theta(-X))$ and, applying eq. (\ref{epsilon_corr}), we find that $c_{bm}\geq 0.0315$.

\section{Upper bounds on $c_{bm}$}
\label{up_bm_app}

As explained in the main text, the problem of upper bounding $c_{bm}$ is equivalent to that of lower bounding the bottom of the spectrum of the operator $A$ defined through $\tr(\rho A) = \int_{\R^2} \theta(x)\theta(p) W_\rho(x,p)$, constrained to the space ${\cal Q}$ of wave-functions $\ket{\psi}$ satisfying $\Theta(X+P)\ket{\psi}=\ket{\psi}$.

Restricted to this space, $A = A + B$, for any operator $B$ that integrates a Wigner function on some region $R \subset \{(x,p) \in \R^2 : x+p\leq 0\}$. Therefore, 
\[
\sup_B\inf\{\lambda : \lambda \in \sigma(A+B)\}  \leq \inf\{\lambda : \lambda \in \sigma(A|_{{\cal Q}})\}.
\]
Unfortunately, computing integrals of Wigner functions on arbitrary regions of phase space is arbitrarily complicated, so we must restrict ourselves to tractable regions.

Define $R_k := \{(x,p) \in \R^2 : xp \geq k, x\leq 0, p \leq 0\}$. Such hyperbolic regions are invariant under the action of the dilation group $e^{it(XP+PX)}$, and it turns out that the operator $B_k$ representing integration over $R_k$ can be block-diagonalized in a basis $\{\ket{\eta}_{+},\ket{\eta}_{-}\}_\eta$ of dilation eigenvectors, exactly like Werner does for $B_0$ in \cite{Werner_1988}. The hyperbolic regions were independently considered in full generality in \cite{wood}, where the spectrum is also numerically computed. The result can only be expressed as follows in terms of integrals which do not have an analytical expression, as far as we are aware:
\[
B_k = \int_{-\infty}^\infty d\eta \sum_{m,n=+,-}K^k_{mn}(\eta)\ket{\eta}_{m}\bra{\eta}_n,
\]
where
\begin{align*}
K_{++}(\eta) &:=0,\\
K_{-+}^k(\eta) &:= \frac{1}{2\pi i}\int_{0}^\infty dx e^{i\eta x}\frac{e^{-2ki\text{Coth}(x)}}{\text{Cosh}(x)},\\
K_{+-}^k(\eta) &:=\overline{K_{+-}^k(\eta)},\\
K_{--}^k(\eta) &:= \lim_{\varepsilon\rightarrow0} \int_{-\infty}^\infty dx e^{i\eta x}\frac{e^{-2ki\text{Tanh}(x)}}{\varepsilon \text{Cosh}(x)+2i\text{Sinh}(x)}
\end{align*}

As proven in \cite{Werner_1988}, the operator $A$ is block-diagonalized by the same unitary transformation. Setting $B=\sum_k b_kB_k$, we thus have that the bottom of the spectrum of $\tilde{A}:=A+B$ equals $\inf_\eta\{\lambda_{\min}(\tilde{A}(\eta))\}$, where $\lambda_{\min}(Z)$ denotes the minimum eigenvalue of the matrix $Z$ and $\{\tilde{A}(\eta):\eta\in\R\}$ is a one-parameter of $2\times 2$ matrices.

In this regard, the best combination we could find before the integrals defining the entries of $\tilde{A}(\eta)$ became too numerically unstable to be reliable is
\[
\tilde{A}:= A + 0.7673B_0 - 0.8767B_{0.1} + 0.09895B_{0.5},
\]
whose spectrum as a function of $\eta$ is shown in Figure \ref{fig: spectrum}.


\begin{figure}
    \centering
    \includegraphics[scale = 0.7]{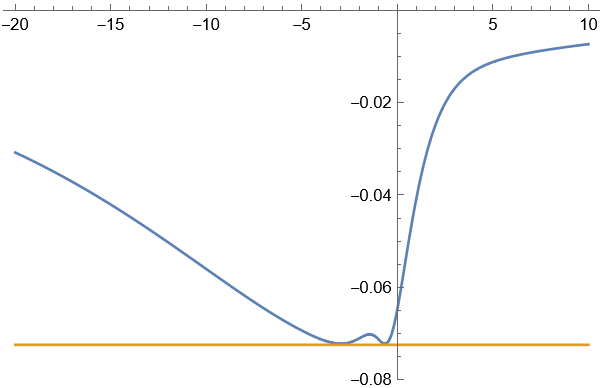}
    \caption{The bottom of the spectrum of the operator $\tilde{A}$, restricted to each two-dimensional subspace $\mbox{span}\{\ket{\eta}_{+},\ket{\eta}_{-}\}$. The horizontal line is $-0.0725$. This spectrum was computed by numerically integrating with Mathematica, taking $\varepsilon = 0.001$ rather than a limit.}
    \label{fig: spectrum}
\end{figure}

\section{Quantum rockets}
\label{rockets_app}

Under the assumption that the map (9) describes fuel combustion, consider a rocket that, most of the time, freely propagates through space, except at times $0=t_1<t_2<t_3<...<t_N$, when the rocket burns fuel instantaneously. We assume that, initially, the state of the rocket's center of mass is $\rho\in {\cal S}([0,l])$, with canonical operators $X^{(0)}, P^{(0)}$. At time $t_j$, the rocket burns a fuel mass $m_j$, hence reducing its mass to $M_j=M-\sum_{k=1}^jm_j$, and experiencing a transformation $\rho\to\Upsilon(\rho;\sigma^{(j)},M_{j-1},m_j)$, where $\sigma^{(j)}\in \S \left(\left[-\frac{\lambda}{2},\frac{\lambda}{2}\right]\right)$ of the fuel in the rocket's reference frame, with canonical operators $X^{(j)}_{REL}, P^{(j)}_{REL}$. Between the times $t_j$ and $t_{j+1}$, the rocket propagates freely and thus its canonical operators $X_R, P_R$ experience the transformation

\be
X_R\to X_R+\frac{t_{j+1}-t_{j}}{M_j}P_R,P_R\to P_R.
\label{free_evo}
\ee

\noindent Call $X_R^{(j)}, P_R^{(j)}$ the canonical operators of the rocket at time $t_j$, just before the new fuel combustion. From eqs. (8), (\ref{free_evo}) it is easy to see that they satisfy the relation

\begin{align}
&X^{(j)}_R=X^{(j-1)}_R-\frac{m_j}{M_j}X^{(j)}_{REL}+\frac{1}{M_j}P_R^{(j-1)}-\frac{1}{M_j-m_j}P^{(j)}_{REL},\nonumber\\
&P^{(j)}_R=\frac{M_j-m_j}{M_j}P_R^{(j-1)}-P^{(j)}_{REL}.
\label{iterative_rel}
\end{align}

Through repeated iteration of (\ref{iterative_rel}), we can express the rocket's final position operator $X_R^{(N)}$ as a linear combination of $X^{(0)}_R,P^{(0)}_R$ and $\{X_{REL}^{(j)}, P_{REL}^{(j)}\}$. That is, for some real vectors $\vec{c}, \vec{d}$, we have $X_R^{(N)}=\vec{c}\cdot\vec{X}+\vec{d}\cdot\vec{P}$, where $\vec{X}=(X_R^{(0)},X_{REL}^{(1)},...)$ and $\vec{P}=(P_R^{(0)},P_{REL}^{(1)},...)$. The probability of detecting the quantum rocket at time $t_N$ in $[a,\infty)$ and its classical counterpart is thus given by

\be
\left\langle\Theta\left(\vec{c}\cdot\vec{X}+\vec{d}\cdot\vec{P}-a\right)\right\rangle_{\rho},
\label{quantum_rocket_formula}
\ee
where $\rho=\rho^{(0)}\otimes\bigotimes_{k=1}^N\sigma^{(k)}$. Since eq. (\ref{iterative_rel}) also holds for classical systems, so does eq. (\ref{quantum_rocket_formula}), when we understand $\rho$ as a product of probability densities. We now consider a classical rocket with the same combustion schedule as the quantum one, and such that the probability densities for the classical moment variables $p_R^{(0)},p^{(1)}_{REL},p^{(2)}_{REL},...$ respectively coincide with those of the states $\rho^{(0)},\sigma^{(1)},\sigma^{(2)},...$. We further assume that the distributions of the initial position of the rocket and the fuel explosions respectively have supports $[0,l]$ and $\left[-\lambda/2,\lambda/2\right]$, just like in the quantum case. Then, the maximum probability of detecting the classical rocket in $[a,\infty)$ at time $t_N$ is
\be
\left\langle\Theta\left(\vec{d}\cdot\vec{P}-(a-L^+)\right)\right\rangle_{\rho},
\label{classical_rocket_formula}
\ee
where
\be
L^+\equiv l\max(0,c_0)+\frac{\lambda}{2}\sum_k|c_k|.
\ee
The maximum advantage $\varphi_R$ of such a quantum rocket is thus the result of maximizing
\be
\left\langle\Theta\left(\vec{c}\cdot\vec{X}+\vec{d}\cdot\vec{P}-a\right)-\Theta\left(\vec{d}\cdot\vec{P}-(a-L^+)\right)\right\rangle_{\rho},
\label{quantum_rocket_advantage}
\ee
over all separable states $\rho=\rho^{(0)}\otimes\bigotimes_{k=1}^N\sigma^{(k)}$ such that $\rho^{(0)}\in{\cal S}[0,l]$, $\sigma^{(j)}\in \S \left[-\lambda/2,\lambda/2\right]$, for $j=1,...,N$. Call $\rho^\star$ the corresponding maximizer (if the maximizer does not exist, then the following argument still carries through if the average value of (\ref{quantum_rocket_advantage}) with $\rho=\rho^\star$ is $\phi_R-\epsilon$).

Now, consider the commutator $[\vec{c}\cdot\vec{X},\vec{d}\cdot\vec{P}]=i\beta$, and assume that $\beta> 0$. Then, 
\[
X\equiv \vec{c}\cdot\vec{X}\to_S X, \; P\equiv \frac{1}{\beta}\vec{d}\cdot\vec{P},
\]
 are canonically conjugated operators. Let $\tilde{\rho}$ be the result of tracing out all degrees of freedom of $\rho^\star$, but that corresponding to $X,P$. Then we have that 
\be
\varphi_R=\left\langle\Theta\left(X+\beta P-a\right)-\Theta\left(\beta P-(a-L^+)\right)\right\rangle_{\tilde{\rho}},
\ee
with $\tilde{\rho}\in{\cal S}([L^{-},L^+])$, with 
\be
L^{-}:= l\min(0,c_0)-\frac{\lambda}{2}\sum_k|c_k|.
\ee
Hence we end up computing $\varphi_F$ under an extra restriction on the quantum states to be optimized. Through the metaplectic transformation $X\to X-L^-, P\to P$, we can map this problem to an optimization over the operator
\be
\varphi_R=\left\langle\Theta\left(X+\beta P-a'\right)-\Theta\left(\beta P-(a'-L)\right)\right\rangle_{\tilde{\rho}},
\ee
\noindent over a constrained set of quantum states contained in $\S([0,L])$, with $L=L^+-L^-$, $a'=a-L^-$. This means that $\varphi_R\leq \varphi(L^2/\beta)\leq\varphi(\infty)\approx 0.038452$. 

If $\beta<0$, we apply the time-reversal anti-unitary operator $X_{R}^{(0)}\to X_{R}^{(0)}$, $X_{REL}^{(j)}\to X_{REL}^{(j)}$, $P_{R}^{(0)}\to -P_{R}^{(0)}$, $P_{REL}^{(j)}\to P_{REL}^{(j)}$ on the operator of eq. (\ref{quantum_rocket_advantage}). This transformation does not affect the spatial support or separability of $\rho$, but effectively changes the sign of $\vec{d}$; and thus, of the commutator, in which case the argument above carries through.

Finally, if $\beta=0$, then $\vec{c}\cdot\vec{X},\vec{d}\cdot\vec{P}$ are commuting operators, in which case eq. (\ref{quantum_rocket_advantage}) cannot have a value greater than $0$.

The final conclusion is that a quantum rocket cannot be more advantageous than a quantum projectile.

\section{The restricted projectile scenario}
\label{variant_app}
\subsection{Computation of $p_c^\star(\mu,\nu)$}
\label{p_c_comput}

In this section, we solve the following problem.
\begin{problem}
Let $\mu(x), \nu(p)$ be the position and momentum distributions of a classical particle of mass $M$. What is the maximum probability $p^\star_c(\mu,\nu)$ that, after time $\Delta T$, we find the particle in the region $[a,\infty)$?
\label{problem_whole_distr}
\end{problem}

In classical mechanics, such a particle is described by its phase space distribution $W(x,p)$ constrained to have position and momentum marginals $\mu(x)$ and $\nu(p)$. Though our notation is similar to Wigner functions, we emphasize that here $W(x,p)\geq0$ is a valid probability distribution. The problem is maximizing  ${\text{Prob}(X+P\Delta T/M\geq a)}$ over random variables $(X,P)$ jointly distributed according to $W(x,p)$ with given marginals.

A discretized version of the problem is maximizing the fraction of pairs $(x_i,p_i)_{i=1}^N$, sampled from $W(x,p)$, satisfying $x_i+p_i\Delta T/M\geq a$. That is to find a permutation $\sigma\in S_N$ maximizing the fraction above. Let $y_i := p_i\Delta T/M$, the initial momentum distribution becomes $\tilde{\nu}(y)dy=\frac{M}{\Delta T}\nu\left(\frac{M}{\Delta T}y\right)dy$.

\begin{lemma}
Given $\vec{x}=(x_i)_{i=1}^N, \vec{y}=(y_j)_{j=1}^N$, let the indices $\hat{i},\hat{j}\in\{1,...,N\}$ be
\begin{align}
\hat{i}&:=\arg\min\{X_i:\exists j, X_i+Y_j\geq a\},\nonumber\\
\hat{j}&:=\arg\min\{Y_j: X_{\hat{i}}+Y_j\geq a\}.
\end{align}
Then, there exists an optimal permutation $\sigma\in S_N$ such that $\sigma(\hat{i})=\hat{j}$.
\end{lemma}
\begin{proof}
Let $\sigma$ be a permutation maximizing the number of pairs $(x_i,y_{\sigma(i)})$ satisfying $x_i+y_{\sigma(i)}\geq a$. If $\sigma(\hat{i})=\hat{j}$, then the lemma holds with the permutation $\sigma$. If $\sigma(\hat{i})\not=\hat{j}$, then the lemma holds with the permutation $\sigma'\in S_n$, defined by
\begin{align}
&\sigma'(\hat{i})=\hat{j}, \sigma'(\sigma^{-1}(\hat{j}))=\sigma(\hat{i}),\nonumber\\
&\sigma'(i)=\sigma(i), \forall i\not=\hat{i}, \sigma^{-1}(\hat{j}).
\end{align}
Indeed, note that the transition $\sigma\to\sigma'$ just affects the pairs
\begin{equation}
    (x_{\hat{i}},y_{\sigma(\hat{i})}),(x_{\sigma^{-1}(\hat{j})},y_{\hat{j}})\label{pairs}.
\end{equation}

By definition of $\hat{i},\hat{j}$, if the second pair adds up to $a$ or more, then $x_{\sigma^{-1}(\hat{j})}\geq x_{\hat{i}}, y_{\sigma(\hat{i})}\geq y_{\hat{j}}$. In that case, the transition will make both final pairs add up to $a$ or more. On the contrary, if the second pair adds up to a number lower than $a$, this means that, at most, just the first pair was satisfying the sum condition. After the transition, though, the pair $(x_{\hat{i}},y_{\hat{j}})$ satisfies it by definition. So once again the transition cannot decrease the number of pairs satisfying the sum condition.

It follows that $\sigma'$ is optimal if $\sigma$ is optimal. Since $\sigma'(\hat{i})=\hat{j}$, the conditions of the lemma are satisfied.
\end{proof}

The lemma suggests a simple algorithm to arrive at an optimal permutation $\sigma$, given the vectors $\vec{x}, \vec{y}$. Namely,
\begin{enumerate}
    \item Define ${\cal I}=\{1,...,N\}, {\cal J}=\{1,...,N\}$, ${\cal S}=\emptyset$.
    \item Find $\hat{i}, \hat{j}$ such that
    \begin{align}
    &\hat{i}=\arg\min\{x_i:i\in {\cal I},\exists j\in {\cal J}, x_i+y_j\geq a\},\nonumber\\
    &\hat{j}=\arg\min\{y_j:j\in {\cal J}, x_{\hat{i}}+y_j\geq a\}.
    \end{align}
\noindent If no such indices exist, return any permutation $\sigma\in S_{N}$ with ${\cal S}\subset\{(i, \sigma(i)):i\in\{1,...,N\}\}$ and halt.
    \item
    Redefine ${\cal S}\leftarrow {\cal S}\cup \{(\hat{i},\hat{j})\}$, ${\cal I}\leftarrow {\cal I}\setminus \{\hat{i}\}, {\cal J}\leftarrow {\cal J}\setminus \{\hat{j}\}$ and go to 2.
\end{enumerate}

To solve the problem posed at the beginning of the section, we just need to apply the algorithm above in a scenario where $N\gg 1$. In that limit, the quantities $N\mu(x)dx, N\tilde{\nu}(y)dy$ approximate the number of entries of $\vec{x},\vec{y}$ with values in $(x, x+dx]$ and $(y,y+dy]$. 

Suppose that we have already paired or discarded all entries of $\vec{x}$ with value smaller than or equal to $x$. Call $S$ the number of pairs already established and $Q$, the number of elements of $\vec{y}$ which are still unpaired and are greater than or equal to $a-x$. Following the algorithm, we need to check how many of the $N\mu(x)dx$ points with value in $(x,x+dx]$ we can pair with the remaining entries of $\vec{y}$. The only possible candidates in $\vec{y}$ are either among the entries already counted in $Q$ or among the $N\tilde{\nu}(a-x)dx$ entries with values in $(a-x,a-(x+dx)]$. If $Q>0$, then all the entries with values in $(x,x+dx]$ can be paired, i.e., $dS=N\mu(x)dx$. In that case, after removing those, the remaining entries of vector $\vec{x}$ are in the interval $(x+dx,\infty)$. Also, the number of unpaired elements of $\vec{y}$ greater than $a-(x+dx)$ are $Q+dQ$, with $dQ=Ndx(\tilde{\nu}(a-x)-\mu(x))$. If $Q=0$, then there are two possibilities: (1) $\tilde{\nu}(a-x)\geq \mu(x)$, in which case $N\mu(x)dx$ entries can be paired, and so $dS=N\mu(x)dx$, $dQ=Ndx(\tilde{\nu}(a-x)-\mu(x))$; (2) $\tilde{\nu}(a-x)< \mu(x)$, in which case just $N\tilde{\nu}(a-x)dx$ can be paired, and so $dS=N\tilde{\nu}(a-x)dx$, $dQ=0$. Defining $q\equiv \frac{Q}{N}$, $s\equiv \frac{S}{N}$, we hence have that the functions $q(x),s(x)$ follow the system of differential equations
\begin{align}
    \frac{ds}{dx}=&\Theta(q)\mu(x)+(1-\Theta(q))\min\left(\mu(x),\tilde{\nu}(a-x)\right),\nonumber\\
    \frac{dq}{dx}=&\Theta(q)(\tilde{\nu}(a-x)-\mu(x))+\nonumber\\
    &(1-\Theta(q))\max\left(\tilde{\nu}(a-x)-\mu(x),0\right).
    \label{ODE}
\end{align}

Call $(s(x),q(x))$ the solution of the system of ordinary differential equations (\ref{ODE}) with the boundary condition $s(-\infty)=q(-\infty)=0$. From all the above it follows that the solution of Problem \ref{problem_whole_distr} is $p_c^\star(\mu,\nu)=s(\infty)$.

\subsection{The gradient of $p_c^\star(\mu,\nu)$}
\label{grad_p_c_comput}
Suppose that the distributions $\mu,\nu$ depend on one parameter $\lambda$, i.e., $\mu=\mu(x;\lambda),\nu=\nu(x;\lambda)$. We wonder how much $p_c^\star(\lambda)= p_c^\star(\mu(\bullet;\lambda),\nu(\bullet;\lambda))$ differs from $p_c^\star(\lambda+\delta\lambda)$, with $\delta\lambda\ll 1$. Let us assume that the roots of $q(x;\lambda)$ can be expressed as $\bigcup_{i=1}^N[x^-_i,x^+_{i}]$, with $x^+_i<x^-_{i+1}$, for all $i$. Then, 
\begin{align}
p^\star_c(\lambda)=\sum_i\int_{x_i^-}^{x_{i}^+}f^-(x;\lambda) dx+
\sum_i\int_{x_i^+}^{x_{i+1}^-}f^+(x;\lambda) dx,
\end{align}
\noindent where $x_{N+1}^{-}=\infty$ if $x_N^{+}<\infty$, and
\begin{align}
&f^+(x;\lambda):=\mu(x;\lambda),\nonumber\\
&f^-(x;\lambda):=\min(\mu(x;\lambda),\tilde{\nu}(a-x;\lambda)).
\label{efes_beer}
\end{align} 
An increment of $\lambda$ will thus have two effects on $p^\star_c(\lambda)$. On one hand, the functions $f^+,f^-$ will respectively change by the amounts $\frac{\partial}{\partial \lambda}f^+\delta\lambda$, $\frac{\partial}{\partial \lambda}f^-\delta\lambda$. On the other hand, the set of points $x$ where $q(x,\lambda+\delta\lambda)$ vanishes will change. Assuming that the kernel of $q(\bullet,\lambda+\delta\lambda)$ is of the form $\bigcup_i[x^-_i+\delta x^-_i,x^+_{i}+\delta x^+_i]$, then we have that
\begin{widetext}
\begin{align}
p^\star_c(\lambda+\delta\lambda)-p^\star_c(\lambda)&\approx\delta \lambda\sum_i\int_{x_i^+}^{x_{i+1}^-}\frac{\partial}{\partial \lambda}f^+(x;\lambda) dx + \delta \lambda\sum_i\int_{x_i^-}^{x_{i}^+}\frac{\partial}{\partial \lambda}f^-(x;\lambda) dx \nonumber\\ 
&+ \sum_i\delta x_i^-\left(f^+(x_i^-;\lambda)-f^-(x_i^-;\lambda)\right)+\sum_i \delta x_i^+\left(f^-(x_i^-;\lambda)-f^+(x_i^-;\lambda)\right).
\label{variation}
\end{align}
\end{widetext}

From eqs. (\ref{ODE}), and, assuming that $\mu,\nu$ are smooth, we have that 
\begin{equation}
\tilde{\nu}(a-x_i^+;\lambda)-\mu(x_i^+;\lambda)=0.
\label{zero_diff}
\end{equation}

Indeed, if this quantity were negative, then $q$ would have remained zero; and, if it were positive, then there would exist $x_i^-<x<x_i^+$ such that $\tilde{\nu}(a-x;\lambda)-\mu(x;\lambda)=0$, and $q$ would have lifted itself from zero at $x$ instead of $x_i^+$. By (\ref{efes_beer}), this implies that the last term of eq. (\ref{variation}) vanishes.

As for the second-to-last term, from eq. (\ref{ODE}) it follows that, for $x=x_i^-+\delta x$, $x<x_i^-+\delta x_i^-$,

\begin{align}
q(x;\lambda+\delta\lambda)\approx &\delta \lambda\int_{x_{i-1}^+}^{x_i^-}\frac{\partial}{\partial \lambda}(\tilde{\nu}(a-y;\lambda)-\mu(y;\lambda))dy+\nonumber\\
&\delta x (\tilde{\nu}(a-x_i^-;\lambda)-\mu(x_i^-;\lambda)),
\end{align}
\noindent where we have used the identities $q(x_{i-1}^+,\lambda)=0$ and (\ref{zero_diff}). Equaling this last equation to zero, we find that

\begin{equation}
\delta x_i^-\approx\frac{\delta \lambda\int_{x_{i-1}^+}^{x_i^-}\frac{\partial}{\partial \lambda}(\tilde{\nu}(a-y;\lambda)-\mu(y;\lambda))dy}{\tilde{\nu}(a-x_i^-;\lambda)-\mu(x_i^-;\lambda)}.
\label{zeros_q}
\end{equation}

Define
\begin{align}
&g^+(x;\lambda):=\tilde{\nu}(a-x;\lambda)-\mu(x;\lambda),\nonumber\\
&g^-(x;\lambda):=\max\left(\tilde{\nu}(a-x;\lambda)-\mu(x;\lambda),0\right).
\end{align}
\noindent From eqs. (\ref{variation}), (\ref{zeros_q}), we find that $\frac{\partial}{\partial \lambda} p^\star_c(\lambda)=s_\lambda(\infty)$, where the function $s_\lambda(x)$ and the auxiliary function $q_\lambda(x)$ evolve according to the system of ordinary differential equations
\begin{align}
\frac{dq}{dx}=&\Theta(q)g^+(x;\lambda)+(1-\Theta(q))g^-(x;\lambda),\nonumber\\
\frac{ds_\lambda}{dx}=&\Theta(q)\frac{\partial}{\partial\lambda} f^+(x;\lambda)+(1-\Theta(q))\frac{\partial}{\partial\lambda} f^-(x;\lambda),\nonumber\\
\frac{dq_\lambda}{dx}=&\Theta(q)\frac{\partial}{\partial\lambda} g^+(x;\lambda)+(1-\Theta(q))\frac{\partial}{\partial\lambda} g^-(x;\lambda),
\label{gradients}
\end{align}
\noindent for $x\not\in \{x_i^-\}_i$, and otherwise are updated as indicated below (notice the update order):
\begin{align}
&s_\lambda(x)\to s_\lambda(x)+\frac{q_\lambda(x)}{\mu(x;\lambda)-\tilde{\nu}(a-x;\lambda)},\nonumber\\
&q_\lambda(x)\to 0,
\label{discont_update}
\end{align}
\noindent Here, the boundary conditions are $q(-\infty)=s_\lambda(-\infty)=q_\lambda(-\infty)=0$. Note that the first line of (\ref{gradients}) is the same as the second line of (\ref{ODE}). Hence it is advisable to run the algorithm to find $s(\infty)$ and its differential $s_\lambda(\infty)$ at the same time. Also, notice that the algorithm sometimes requires us to differentiate a non-differentiable function, such as $f^{-}(x;\lambda)=\min(\mu(x;\lambda),\nu(x;\lambda))$. In that case, we define $\partial f^{-}(x;\lambda)/\partial\lambda$ as $\partial \mu(x;\lambda)/\partial\lambda$, if $\mu(x;\lambda)<\nu(x;\lambda)$; or $\partial \nu(x;\lambda)/\partial\lambda$, otherwise. In doing so, we are implicitly assuming that the equation $\mu(x;\lambda)=\nu(x;\lambda)$ has a countable number of roots in $x$. The definition of $\partial g^{-}(x;\lambda)/\partial\lambda$ is analogous.

\subsection{Maximizing $\mathbb{W}(\rho)$}
Since the maximum quantum advantage is independent of the parameters $M,\Delta T, a$, from now on we take $M=\Delta T=1$, $a=0$ in $\mathbb{W}(\rho)$ that is
\begin{equation}
{\cal W}(\rho)\equiv\langle\Theta(X+P)\rangle_\rho-p_c^\star(\rho).
\end{equation}

We will only perform projected gradient ascend in the subspace $\H_N=\mbox{span}\{\ket{n}:n=0,\dots,N\}$, noting that we can get better achievable {\em lower bounds} with increasing number of iterations and increasing $N$. For a learning rate $\epsilon$, each iteration updates the state according to
\begin{equation}\label{proj_grad}
\rho^{k+1}={\cal P}(\rho^{k}+\epsilon\nabla_{\rho}{\cal W}(\rho)),
\end{equation}
with the projection ${\cal P}$ ensuring a valid density matrix. It remains to compute various quantities above.

Firstly, for any matrix $M$, the projection ${{\cal P}(M)=\text{argmin}_Z\|Z-M\|_2}$ to the set of density matrix can be cast as a semidefinite program \cite{sdp}: 
\begin{align}
\min_{Z,\Lambda} &\quad\tr(\Lambda)\nonumber\\
\mbox{s.t. } & \left(\begin{array}{cc}\Lambda&M-Z\\M-Z&\id\end{array}\right)\geq 0,\nonumber\\
&Z\geq0,\tr(Z)=1.
\end{align}
To solve this program, we used the MATLAB package YALMIP \cite{yalmip} in combination with the semidefinite programming solver MOSEK \cite{mosek}. 

Next, $\nabla_\rho{\cal W}(\rho)=\nabla_\rho\tr(\rho \Theta(X+P))-\nabla_\rho s(\infty)$. Write
\begin{equation}
\rho=\sum_{m,n=0}^N(\mbox{Re}(\rho_{mn})+ i\mbox{Im}(\rho_{mn}))\ket{m}\bra{n} 
\end{equation}
and let $\rho^R_{mn}=\Re(\rho_{mn}),\rho^I_{mn}=\Im(\rho_{mn})$ be our real variables to be optimized. Then the first term
\begin{align}
    \nabla_\rho\tr(\rho \Theta(X+P)) = \nabla_\rho\tr(\rho{\cal O}_N) = {\cal O}_N
\end{align}
where ${\cal O}_N$ is the $(N+1)\times(N+1)$ matrix with entries given by \eqref{theta_interm} for $\phi=\pi/4$.

Finally, the last term $\nabla_\rho s(\infty)=[s_{\rho_{mn}}(\infty)]_{m,n=0}^N$ consists of individual gradients $s_{\rho_{mn}}(\infty)$, in the notation of the previous section, receiving contribution from the gradient with respect to real parameters $\rho^R_{mn},\rho^I_{mn}$. Call $\mathbb{M}_{n}(\C)$ the set of $n\times n$ complex matrices. From \eqref{gradients} and \eqref{discont_update}, we have $\nabla_\rho s(\infty)=s_\rho(\infty)$ for  ${s_\rho:\R\to\mathbb{M}_{N+1}(\C)}$ the solution of 
\begin{align}
\frac{dq}{dx} &= \Theta(q)g^+(x;\rho)+(1-\Theta(q))g^-(x;\rho),\nonumber\\
\frac{ds_\rho}{dx} &= \Theta(q)\nabla_\rho f^+(x;\rho)+(1-\Theta(q))\nabla_\rho f^-(x;\rho),\nonumber\\
\frac{dq_\rho}{dx} &= \Theta(q)\nabla_\rho g^+(x;\rho)+(1-\Theta(q))\nabla_\rho g^-(x;\rho),
\end{align}
\noindent for $x\not\in \{x_i^-\}_i$, and otherwise are updated as
\begin{align}
&s_\rho(x)\to s_\rho(x)+\frac{q_\lambda(x)}{\mu(x;\lambda)-\tilde{\nu}(a-x;\lambda)},\nonumber\\
&q_\rho(x)\to 0_{N+1\times N+1}.
\end{align}

This system of differential equations contains auxiliary functions $q:\R\to\R$ and $q_\rho:\R\to \mathbb{M}_{N+1}(\C)$ to be solved, as well as data functions $f^+,f^-,g^+,g^-:\R\to\R$ depending on parameters $\rho$ and their gradients. Recalling their definitions, we get that their gradients depends on
\begin{align}
\nabla_\rho\mu(x;\rho) &= \ketbra{v(x)}{v(x)},\nonumber\\
\nabla_\rho\nu(p;\rho) &= \ketbra{w(p)}{w(p)},
\end{align}
where $\ket{v(x)}, \ket{w(x)}$ are the $N+1$-dimensional vectors with entries $v(x)_n=\braket{n}{x}$, $w(x)_n=\braket{n}{p}$ for $n=0,...,N$ and 
\begin{align}
\braket{x}{n} &= \frac{1}{\sqrt{2^nn!}\pi^{1/4}}e^{-\frac{x^2}{2}}H_n(x),\nonumber\\
\braket{p}{n} &= \frac{(-i)^n}{\sqrt{2^nn!}\pi^{1/4}}e^{-\frac{p^2}{2}}H_n(p),
\end{align}
where $H_n(z)$ denotes the Hermite polynomial of degree $n$, defined by
\begin{equation}
H_n(z):=(-1)^ne^{z^2}\frac{d^n}{dz^n}e^{-z^2}.
\end{equation}
In summary, we have completely specify the data defining the system of differential equations, as well as the computation leading to the update \eqref{proj_grad}.

To solve the above system of ordinary differential equations, as well as \eqref{ODE}, we used the Euler explicit method with step $\Delta x=0.0001$. For practical reasons, we had to limit the range of possible values of $x$, which we set to be $[-40,40]$.

\end{appendix}

\end{document}